\documentclass[%
 reprint,
 amsmath,amssymb,
 aps, physrev,
]{revtex4-2}
\usepackage{physics}
\usepackage{graphicx}
\usepackage{dcolumn}
\usepackage{bm}
\usepackage{amsmath,amssymb,amsthm,mathtools}
\usepackage[hidelinks]{hyperref}
\usepackage[colorinlistoftodos]{todonotes}

\newtheorem{definition}{Definition}
\newtheorem{theorem}{Theorem}
\newtheorem{lemma}{Lemma}
\newtheorem{fact}{Fact}
\newtheorem{corollary}{Corollary}

\usepackage[ruled,vlined,linesnumbered]{algorithm2e}
\DontPrintSemicolon
\SetKwInput{KwIn}{Input}
\SetKwInput{KwOut}{Output}
\SetKwInput{KwPromise}{Promise}

\begin{document}

\preprint{APS/123-QED}

\title{\textbf{Single-copy stabilizer learning: average case and worst case}}

\author{Gyungmin Cho}
 \email{km950501@snu.ac.kr}
 \affiliation{Department of Physics and Astronomy, and Institute of Applied Physics, Seoul National University, Seoul, Korea}
\author{Dohun Kim}%
 \email{dohunkim@snu.ac.kr}
 \affiliation{Department of Physics and Astronomy, and Institute of Applied Physics, Seoul National University, Seoul, Korea}%

\date{\today}

\begin{abstract}
We study single-copy stabilizer learning, the problem of identifying a stabilizer group
of dimension $n-t$ from an $n$-qubit quantum state $\rho$.
We obtain two complementary results.
First, in the average case, logarithmic-depth local Clifford circuits suffice to efficiently learn
almost all stabilizer groups with $t=O(\log n)$, instead of the linear-depth measurements required in previous approaches. We support this result with numerical simulations for systems of up to 100 qubits. Second, we show that, in the worst case, any adaptive single-copy measurement scheme requires a number of samples that scales exponentially in $t$.
Together with existing results on two-copy learning, our findings suggest that, for large $t$,
identifying Pauli symmetries of a quantum system exhibits a quantum advantage in the learning setting.
\end{abstract}

\maketitle


\section{Introduction}\label{section:intro}
Stabilizer states are a central class of quantum states, defined as the common eigenspace of an abelian subgroup of the Pauli group. An $n$-qubit stabilizer state is uniquely specified by $n$ independent commuting Pauli operators. Stabilizer states play an important role across quantum information science and many-body physics. In quantum information, most quantum error-correcting codes admit a stabilizer description, making these states fundamental for fault-tolerant quantum computation~\cite{gottesman1997stabilizer}. In many-body physics, they also provide canonical examples of topological order, such as the logical states of the toric code~\cite{kitaev2003fault}.

$t$-doped stabilizer states provide a simple interpolation between stabilizer states and more general quantum states. Any stabilizer state can be prepared from $\ket{0^n}$ by a Clifford circuit generated by \{CNOT, Hadamard, Phase\}. If a small number of non-Clifford gates, such as $T$ gates, are inserted, some of the Pauli stabilizer structure is typically lost. The resulting state is no longer stabilized by $n$ independent commuting Pauli operators, but may still retain a large abelian Pauli subgroup. In this paper, we refer to such a state as a \emph{$t$-doped stabilizer state} when the remaining Pauli stabilizer subgroup has dimension $n-t$.

There are several reasons to study such states. In quantum information, Clifford-generated ensembles form a state $3$-design but not a state $4$-design~\cite{kueng2015qubit}, so accessing fourth or higher moments requires some non-Clifford resource. This leads to circuits consisting of Clifford gates together with a small number of non-Clifford gates~\cite{haferkamp2023efficient}, such as $T$ gates, and hence to $t$-doped stabilizer states. In many-body physics, Pauli symmetry subgroups appear in settings such as symmetry-protected topological phases~\cite{zeng2015quantum}, and related Pauli stabilizer structure can also emerge effectively in low-energy sectors even when it is not explicit at the Hamiltonian level~\cite{brell2011toric, kitaev2006anyons}.

As quantum devices continue to improve, they enable the preparation and study of increasingly complex quantum states that are difficult to analyze by direct classical means~\cite{evered2025probing, will2025probing}. In such situations, identifying the symmetries of the underlying state provides useful information. 

In particular, for a typical quantum state, the expectation values of all non-identity Pauli operators are exponentially small with high probability~\cite{huang2021information}. Therefore, the presence of a Pauli operator with a large expectation value already constitutes significant information about the state. The challenge, however, is that there are $4^n - 1$ non-identity Pauli operators on $n$ qubits, so determining such a symmetry by a brute-force search would require exponential samples and time.

To circumvent this exponential barrier, previous works introduced Bell difference sampling~\cite{grewal2025efficient, hangleiter2024bell, leone2024learning, gross2021schur, chen2025stabilizer}, which allows one to extract information about the stabilizer groups using poly($n$) samples and time resources. However, these approaches require entangled measurements across two copies of the state [Fig.~\ref{fig:1}(a)]. Moreover, on hardware with limited qubit connectivity, implementing the Bell-measurement circuit may incur a depth overhead that grows with the system size.

To alleviate these hardware constraints, computational difference sampling was introduced in the single-copy setting~\cite{grewal2025efficient}. This removes the need for quantum memory, but requires the application of a random Clifford circuit before measurement, which has circuit depth $O(n)$ in the absence of ancilla qubits~\cite{MaslovRoetteler2018} [Fig.~\ref{fig:1}(b)]. This raises the question of whether one can learn large Pauli symmetry groups using only single-copy measurements while substantially reducing the required circuit depth, say, $O(\log n)$ [Fig.~\ref{fig:1}(c)].

\begin{figure}[t]
    \centering
    \includegraphics[width=1\linewidth, trim=1cm 14cm 12cm 1cm, clip]{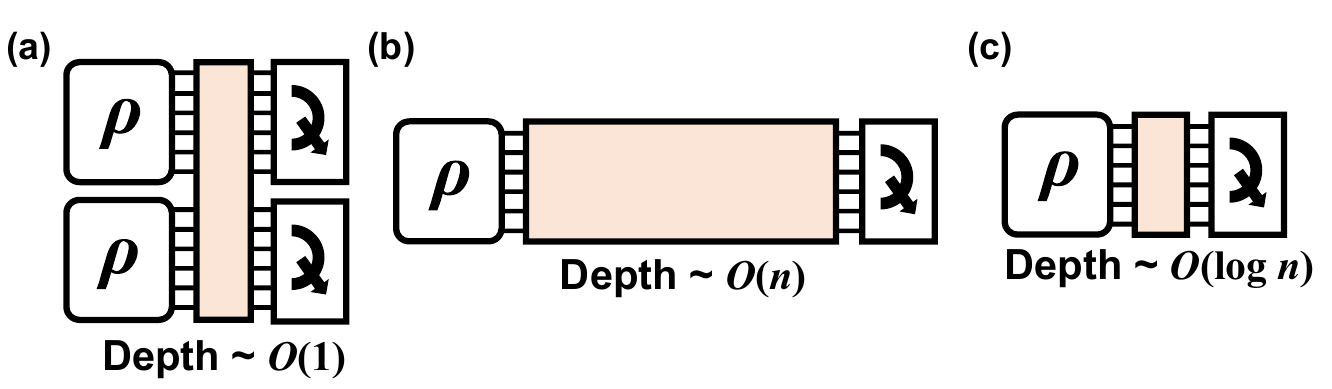}
    \caption{\label{fig:1}
    Comparison of measurement protocols.
    (a) Bell sampling uses two copies of the $n$-qubit state $\rho$ and performs Bell-basis
    measurements. The measurement circuit has depth $O(1)$,
    but the protocol requires $n$-qubit quantum memory.
    (b) Random Clifford measurement uses a single copy of $\rho$, applies a random Clifford circuit $C\sim \mathrm{Cl}(n)$, and then measures in the computational basis. In general,
    implementing $C$ requires circuit depth $O(n)$.
    (c) Random shallow Clifford measurement, as considered in this work, also uses a single copy
    of $\rho$ but replaces $\mathrm{Cl}(n)$ by the block-Clifford ensemble
    $\mathcal{C}_k=\mathrm{Cl}(k)^{\otimes n/k}$ with $k=O(\log n)$, yielding circuit depth
    $O(\log n)$.
    }
\end{figure}

Prior works crucially rely on the fact that the random Clifford group acts transitively on the set of non-identity Pauli operators~\cite{chia2024efficient, grewal2025efficient}. As a result, all such operators lie in a single orbit, and many averaging arguments reduce to analyzing a single representative together with uniform counting over its orbit.

In contrast, for shallow-depth circuits, locality imposes strong constraints on how Pauli operators can propagate. Consequently, the Pauli space no longer forms a single orbit but instead decomposes into multiple, geometry-dependent orbits. In this setting, the standard transitivity-based arguments break down.

To address this, we develop a framework that reduces the problem to cases determined by the relevant symplectic intersection patterns, and then analyzes each case via a refined counting argument. A key technical ingredient is an extension of the counting methods for Lagrangian subspaces~\cite{kueng2015qubit} to general isotropic subspaces.

\section{Preliminaries}
Throughout this work, we ignore global phases of Pauli operators and identify unsigned
$n$-qubit Pauli operators with binary vectors in $V := \mathbb{F}_2^{2n}$.
For $a=(a_x,a_z)\in \mathbb{F}_2^n\times \mathbb{F}_2^n$, we define the corresponding Weyl operator by
$W_a := i^{a_x\cdot a_z}\bigotimes_{j=1}^n X^{(a_x)_j} Z^{(a_z)_j}$, where $a_x\cdot a_z$ in the phase factor is evaluated over the integers. Under this identification, multiplication of unsigned Pauli operators corresponds to addition in $\mathbb{F}_2^{2n}$. For a state $\rho$, we define $\mathrm{Weyl}(\rho) := \{a \in V: \mathrm{tr}(\rho W_a)^2 = 1\}$, and define the \emph{stabilizer dimension} of $\rho$ to be $\dim\mathrm{Weyl}(\rho)$. 

The vector space $V$ is equipped with the canonical symplectic bilinear form. For $a=(a_x,a_z)$ and $b=(b_x,b_z)$ in $V$, we define $[a,b] := a_x \cdot b_z + a_z \cdot b_x \pmod 2$, where~$\cdot$ denotes the standard inner product over $\mathbb{F}_2^n$. In particular, $W_a$ and $W_b$ commute if and only if $[a,b]=0$. A $S\subseteq V$ is called \emph{isotropic} if $[a,b]=0$ for all $a,b\in S$. An isotropic subspace is called \emph{Lagrangian} if it has maximal dimension, i.e., if $\dim(S)=n$. Let $\mathcal{I}_{n-t}(V)$ denote the set of all isotropic subspaces of $V$ of dimension $n-t$. For subspaces $A,B\subseteq V$, we define their sum by $A+B := \{a+b : a\in A,\ b\in B\}$. We will also use the symplectic complement of a subspace $A^\perp := \{x\in V : [x,a]=0 \text{ for all } a\in A\}$. 

We employ the block-Clifford ensemble $\mathcal{C}_k \coloneq \mathrm{Cl}(k)^{\otimes n/k}$, which has appeared previously in randomized measurements~\cite{cho2025entanglement}
and quantum state tomography~\cite{cho2025sample, brandao2020fast}. For a Clifford ensemble $\mathcal{C}$ and a Pauli operator $P$, we define the Pauli collision coefficient~\cite{bertoni2024shallow, cho2025entanglement} by
\begin{equation}
    m_P \coloneq \mathbb{E}_{C\sim\mathcal{C}} [\mathbf{1}\{C(P) \in \mathcal{Z}\}],
\end{equation}
where $C(P)=CPC^{\dagger}$, $\mathcal{Z} = \{I,Z\}^n=\{(0^n,z):z\in \mathbb{F}_2^{n}\}$, and $\mathbf{1}\{\cdot\}$ denotes the indicator function. For $\mathcal{C}_k$ with $k=\Omega(\log n)$, we have [Appendix~\ref{appx:clifford ensemble}]
\begin{equation}
    \min_P m_P = \Theta(2^{-n}).
\end{equation}

Since implementing a unitary in $\mathrm{Cl}(k)$ requires depth $O(k)$~\cite{MaslovRoetteler2018}, choosing $k=O(\log n)$ allows us to realize $\mathcal{C}_k$ using a circuit of depth $O(\log n)$.
\section{Computational difference sampling}

Given an $n$-qubit quantum state $\rho$, computational difference sampling~\cite{grewal2025efficient} consists of
drawing two independent computational-basis outcomes $x,y \in \mathbb{F}_2^n$ from $\rho$
and returning their bitwise sum $x+y \pmod 2$.
We denote the resulting distribution by $r_{\rho}$.

For our analysis, it is convenient to regard the outcome space as $\mathcal{X} \coloneqq \{(x,0^n) : x \in \mathbb{F}_2^n\} \subset \mathbb{F}_2^{2n}$, rather than $\mathbb{F}_2^n$ itself.
Under this identification, for any $a \in \mathbb{F}_2^n$,
\begin{equation}
    r_{\rho}(a,0^n)
    =
    \sum_{x \in \mathbb{F}_2^n}
    \mathrm{tr}(\rho \ket{x}\bra{x})\,
    \mathrm{tr}(\rho \ket{x+a}\bra{x+a}).
\end{equation}
Equivalently, $r_{\rho}(a,0^n)$ is the probability that two independent
computational-basis samples from $\rho$ differ by $a$.
For a subset $A \subseteq \mathcal{X}$, we write $r_{\rho}(A) \coloneqq \sum_{x \in A} r_{\rho}(x)$.

The embedding of the outcome space into $\mathcal{X} \subset \mathbb{F}_2^{2n}$ is used only
to match the symplectic formalism adopted throughout the paper.
In the actual implementation, one may work directly over $\mathbb{F}_2^n$, which yields a
more efficient classical post-processing.

Computational difference sampling itself is well defined for an arbitrary quantum state $\rho$.
However, in prior work, its analysis has typically relied on properties specific to pure states~\cite{grewal2025efficient}.
Here we extend the framework so that it also applies to mixed states.

\section{Learning stabilizer group}
We begin by recalling the definition of stabilizer-group learning, which was introduced by previous works~\cite{grewal2025efficient}.
\begin{definition}[{$(\epsilon,\delta,t)$-stabilizer group learning}]\label{def:learning_task}
    Let $\rho$ be an $n$-qubit state with stabilizer dimension $n-t$.
    An algorithm outputs an isotropic subspace $\widehat S$ such that, with probability at least $1-\delta$,
    \begin{enumerate}
        \item $\widehat S \supseteq \mathrm{Weyl}(\rho)$,
        \item $\mathbb{E}_{x\sim \widehat S}\!\left[\mathrm{tr}(\rho W_x)^2\right] \ge 1-\epsilon$.
    \end{enumerate}
\end{definition}

This problem arises naturally as a subroutine for quantum state tomography (QST) on states with stabilizer dimension $n-t$~\cite{grewal2025efficient, chia2024efficient}.
Indeed, once such a subspace $\widehat S$ is found, one can separate the stabilizer and non-stabilizer parts of the state. It then suffices to perform QST only on the remaining $t$ qubits, thereby avoiding an exponential dependence on $n$. More recently, this perspective has also appeared in the study of the state hidden subgroup problem~\cite{bouland2025state, hinsche2025abelian}; when the underlying abelian group is a stabilizer group, the resulting notion coincides with the definition above. 

Our first main result can be stated as follows. A detailed proof is given in the Appendix~\ref{appx:proof}.

\begin{theorem}[Upper bound for learning almost all $\mathrm{Weyl}(\rho)$]\label{thm1}
    Let $\epsilon < 0.25$, and let $\rho$ be an $n$-qubit quantum state with
    $\dim \mathrm{Weyl}(\rho) = n-t$.
    There exists a single-copy, logarithmic-depth algorithm for
    $(\epsilon,\delta,t)$-stabilizer group learning that, for all but an exponentially small fraction of such $\rho$, uses $O(2^t n^3/\epsilon)$ copies of $\rho$, runs in time $O(2^t n^5/\epsilon)$, and succeeds with high probability.
\end{theorem}

This result shows that, for learning almost all $\rho$ with stabilizer dimension $n-t$ with $t=O(\log n)$, it suffices to use only $O(\log n)$-depth circuits while achieving the same asymptotic resource scaling as previous non-adaptive algorithms based on random Clifford circuits of depth $O(n)$~\cite{grewal2025efficient}. Another notable feature is that, whereas prior works mainly focused on the case of pure-state inputs~\cite{grewal2025efficient, chia2024efficient}, our approach also extends to mixed states. Moreover, when the input is restricted to stabilizer states, the sample complexity can be reduced to $O(n^2)$, with a corresponding improvement in the $n$-dependence of the computational cost. This matches the optimal sample complexity known for learning stabilizer states from single-copy measurements~\cite{arunachalam2022optimal}.

\begin{algorithm}[t]\label{alg}
    \caption{Approximating Weyl($\rho$) using single-copy shallow-depth measurements}
    
    \KwIn{$O(n^3 2^t/\epsilon)$ copies of $\rho$ and $\epsilon \in (0,0.25)$}
    \KwPromise{$\dim \mathrm{Weyl}(\rho) = n-t$}
    \KwOut{A subspace $\widehat{S} \supseteq \mathrm{Weyl}(\rho)$ of dimension at least $n-t$ such that $\mathbb{E}_{x\sim \widehat{S}}\mathrm{tr}(\rho W_x)^2 \geq 1-\epsilon$ with high probability}
    
    \SetKwProg{Fn}{function}{}{}
        Take $m=O(n2^t)$\;
        \For{$i=1$ \KwTo $m$}{
            Sample $C_i \sim \mathcal{C}_k$ and define $\rho_i=C_i \rho C_i^{\dagger}$.\;
            Computational difference sample $\rho_i$ to draw $N_S=O(n^2/\epsilon)$ samples from $r_{\rho_i}$.\;
            Let $\widehat{H}_i\subset \mathcal{X}$ be the subspace spanned by the computational difference samples. 
        }
        output $\widehat{S}=\sum_i^{m} C_i^{\dagger}((\widehat{H}_i+\mathcal{Z})^{\perp})$\;
\end{algorithm}

At a high level, the algorithm proceeds in two steps. First, after applying a random Clifford $C$, it extracts the portion of the stabilizer group that becomes visible in the computational basis (line 3 in Algorithm~\ref{alg}). Second, it uses this information to reconstruct the partial stabilizer groups by computational difference sampling (line 4 in Algorithm~\ref{alg}). In our proof, the most technical part lies in the first step. One must show that, for a Clifford ensemble $\mathcal{C}_k$, the resulting intersections with $\mathcal{Z}$ reveal enough information to recover the entire stabilizer group. The key lemma is the following.

\begin{lemma}[Intersection with $\mathcal{Z}$]\label{main:lem:main_inequal}
    Assume that $t=O(\log n)$ and $k=\Omega(\log n)$. Let $C$ be sampled from $\mathcal{C}_k$. Then, for all but a $2^{-\Omega(n)}$ fraction of $S \in \mathcal{I}_{n-t}(V)$, we have
        \begin{equation}
            \Pr_C\!\left(
                C(S \setminus T) \cap \mathcal{Z} \neq \emptyset
            \right)
            \ge \Omega(2^{-t})
        \end{equation}
    for every codimension-one subspace $T \subset S$.
\end{lemma}

Using Lemma~\ref{main:lem:main_inequal}, we can recover each generator by drawing $O(2^t)$ random Clifford $C$. Since the stabilizer group has at most $n-t$ generators, it follows that a total of $m = O(n2^t)$ random unitaries suffices. We summarize this as follows:

\begin{corollary}[Recovery of the Weyl support]\label{main:cor:weyl-support-recovery}
    Let $C_1, C_2, \dots, C_m$ be sampled from $\mathcal{C}_k = \mathrm{Cl}(k)^{\otimes n/k}$,
    and let $\mathrm{Weyl}(\rho) \in \mathcal{I}_{n-t}(V)$, where $m = O(n2^t)$ and
    $k = \Omega(\log n)$. Then, with high probability over the choice of
    $C_1,\dots,C_m$, the following holds for all but an exponentially small fraction of
    $\mathrm{Weyl}(\rho)$:
    \begin{equation}
        \sum_{i=1}^{m} \bigl(\mathrm{Weyl}(\rho) \cap C_i^{\dagger}(\mathcal{Z})\bigr)
        = \mathrm{Weyl}(\rho).
    \end{equation}
\end{corollary}

Corollary~\ref{main:cor:weyl-support-recovery} shows that the sampled Clifford bases are
sufficient to recover $\mathrm{Weyl}(\rho)$. It therefore remains to perform sufficiently
many computational difference samples in each such basis. This is formalized in the following fact.

\begin{fact}[Theorem 9.6 in~\cite{grewal2025efficient}]\label{fact:recover-heavy-subspace-each-basis}
    Let $\epsilon \in (0,0.25)$. For each $i \in [m]$, let $\rho_i := C_i \rho C_i^\dagger$, and let $\widehat{H}_i \subseteq \mathcal{X}$ be the subspace spanned by $O(n^2/\epsilon)$ independent computational difference samples drawn from $r_{\rho_i}$. Then, with high probability, $r_{\rho_i}(\widehat{H}_i) \ge 1 - \epsilon/2n$ simultaneously for all $i \in [m]$.
\end{fact}

By construction, each $\widehat{H}_i$ accounts for a large fraction of the mass of $r_{\rho_i}$. Summing these subspaces therefore yields an approximation $\widehat{S}$ to $\mathrm{Weyl}(\rho)$. We summarize this conclusion in the following Lemma.

\begin{lemma}[Reconstruction via subspace sum]\label{lem:sum_of_subspaces}
    Let $\epsilon \in (0,0.25)$, and let $\widehat{S}\coloneq\sum_{i=1}^{m} C_i^\dagger \bigl( (\widehat{H}_i + \mathcal{Z})^\perp \bigr)$, where, for every $i\in[m]$, $r_{\rho_i}(\widehat{H}_i)\ge 1-\epsilon/2n$. Then, conditioned on the event in Corollary~\ref{main:cor:weyl-support-recovery},
    $\widehat S$ is an isotropic subspace satisfying
    \begin{equation}
        \widehat S\supseteq \mathrm{Weyl}(\rho),
        \quad
        \mathbb{E}_{x\sim \widehat{S}}\!\left[\mathrm{tr}(\rho W_x)^2\right]
        \ge 1-\epsilon.
    \end{equation}
\end{lemma} 

Indeed, the sample complexity in Theorem~\ref{thm1} follows immediately from the two preceding steps.
By Corollary~\ref{main:cor:weyl-support-recovery}, it suffices to sample
$m=O(n2^t)$ measurement bases. For each such basis, by
Fact~\ref{fact:recover-heavy-subspace-each-basis}, we require
$N_S=O(n^2/\epsilon)$ computational difference samples. Hence, the total number of copies becomes $O(2^t n^3/\epsilon)$.

The runtime in Theorem~\ref{thm1} is dominated by the Gaussian elimination and subsequent subspace
operations used to construct $\{\widehat{H}_i\}$ and aggregate them into
$\widehat{S}$, which yield a total runtime of $O(2^t n^5/\epsilon)$.

\section{Worst case}
The algorithm based on our shallow-depth measurements does not work for all inputs $\rho$. A representative counterexample is the $n$-qubit GHZ state. Among the generators of its stabilizer group $S_{\mathrm{GHZ}}$, $n-1$ have weight $2$, while the remaining one has weight $n$. Hence, if we define $T_0=\langle Z_1Z_2,\dots,Z_{n-1}Z_n\rangle=\{(0^n,z):\sum_i z_i \equiv 0\pmod2\}$, then every element of $S_{\mathrm{GHZ}}\setminus T_0$ has weight $n$, making this an especially hard instance for our approach. In this case, for $C\sim\mathcal{C}_k$, one can show that 
\begin{equation}
    \Pr_C\!\left(C(S_{\mathrm{GHZ}} \setminus T_0) \cap \mathcal{Z} \neq \emptyset\right) \leq (2/3)^{n/k}.
\end{equation}
A proof is given in the Appendix~\ref{appx:ghz}. Therefore, for the Clifford ensemble $\mathcal{C}_k$, $k$ must scale at least linearly with $n$. Nevertheless, this obstruction can be bypassed by employing a single round of adaptivity. Indeed, one can first learn an $(n-1)$-dimensional subgroup $T_0 \subset S_{\mathrm{GHZ}}$ without much difficulty. Since $|T_0^{\perp}/T_0|=4$, it then suffices to identify which of the four candidates is correct. This avoids the need to draw exponentially many unitaries. Moreover, this adaptive step still uses only constant-depth circuits, so the overall depth remains $O(\log n)$.

\section{Numerical simulations}
\begin{figure}[t]
    \centering
    \includegraphics[width=1\linewidth, trim=1cm 7.3cm 3cm 1cm, clip]{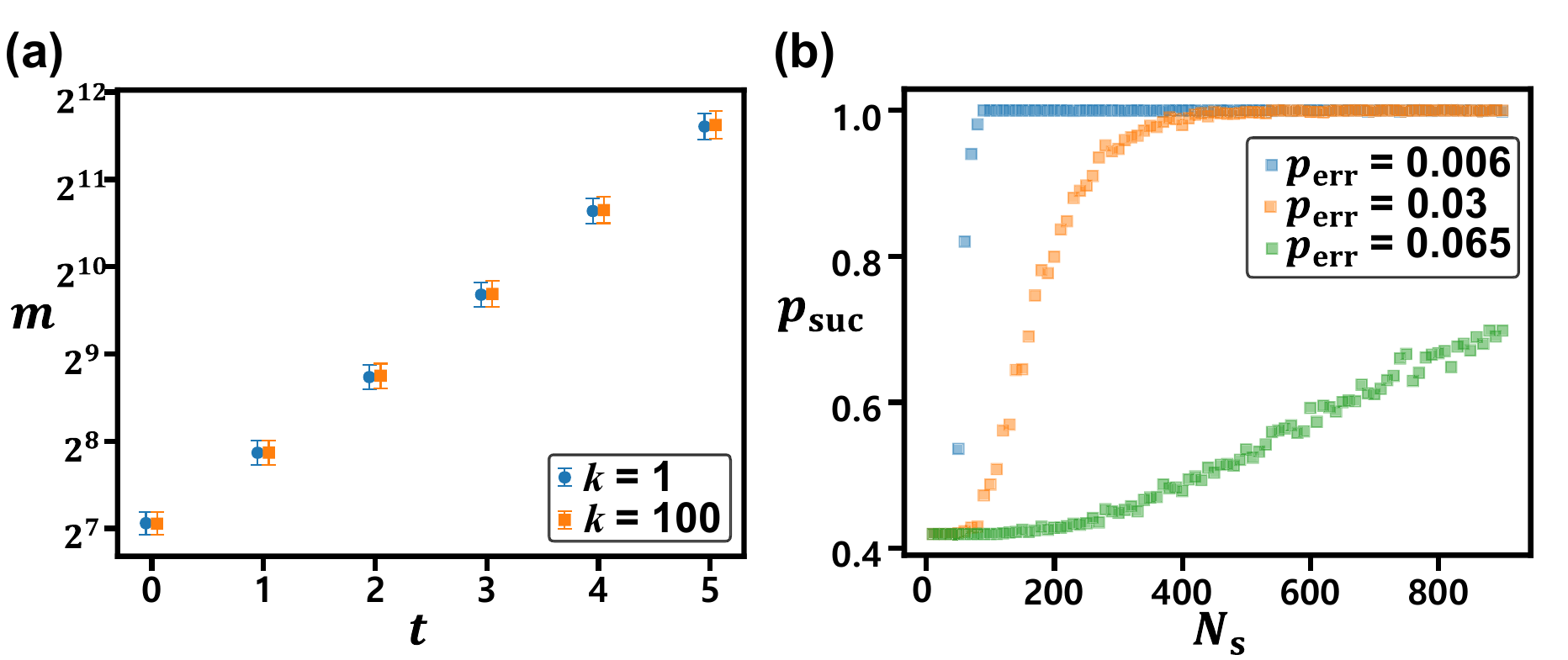}
    \caption{\label{fig:2}
    Numerical simulations.
    (a) Learning $\mathrm{Weyl}(\rho)$ for random instances. We fix $n=100$ and vary $t=0,\dots,5$, where $\dim \mathrm{Weyl}(\rho)=n-t$. For each $t$, we estimate the number $m$ of sampled Clifford circuits required to recover $\mathrm{Weyl}(\rho)$, comparing the Clifford ensembles $\mathcal{C}_k=\mathrm{Cl}(k)^{\otimes n/k}$ for $k=1$ and $k=100$.
    (b) Recovery under noisy computational difference samples. We fix $n=24$, $t=0$, and vary the number $N_S$ of computational difference samples obtained from a single measurement basis. We evaluate the success probability $p_{\mathrm{suc}}$ of recovering the hidden symmetry when each measured bit is flipped independently with probability $p_{\mathrm{err}}$. In both panels, each point shows the mean and standard deviation over $1000$ independent trials. For clarity, error bars are omitted in (b); see Appendix~\ref{appx:numeric} for the full plot.
    }
\end{figure}
Here, we compare the block-Clifford ensemble $\mathcal{C}_k$ with the full random Clifford ensemble $\mathrm{Cl}(n)$ through numerical experiments. Algorithm~\ref{alg} consists of two steps: first, sampling $C \sim \mathcal{C}$, and second, performing computational difference sampling.

If $C^{\dagger}(\mathcal{Z})$ has a non-trivial overlap with $\mathrm{Weyl}(\rho)$, the number of samples required in the computational difference sampling step is the same for both $\mathcal{C}_k$ and $\mathrm{Cl}(n)$. Therefore, we first study how effectively a sampled $C$ can span $\mathrm{Weyl}(\rho)$, by comparing how efficiently the two ensembles reveal generators of Weyl($\rho$).

In Fig.~\ref{fig:2}(a), we fix $n=100$ and vary $t$ from $0$ to $5$. For each ensemble, we record the number $m$ of sampled unitaries required to span $\mathrm{Weyl}(\rho)$. Each data point is averaged over $1000$ trials. While our theoretical analysis uses circuits of depth $O(\log n)$, in the numerics we consider the case $k=1$, corresponding to random Pauli measurements ($\mathcal{C}_1$). As shown in the figure, for all values of $t$, random Pauli measurements and random Clifford measurements require essentially the same number of samples $m$ to recover $\mathrm{Weyl}(\rho)$. Moreover, the required $m$ increases by approximately a factor of two when $t$ increases by one, consistent with the theoretical prediction. Taken together, these numerical results suggest that $k=1$ already suffices for learning almost all stabilizer groups. Additional simulation results for a range of system sizes, as well as for stabilizer states that commonly arise in physics and quantum error correction rather than random instances, are provided in the Appendix~\ref{appx:numeric}.

Reducing the circuit depth from $O(n)$ to $O(\log n)$ directly lowers the cost of gate compilation and implementation. In addition, on noisy quantum devices, fewer gates lead to reduced accumulated errors. Once a suitable $C$ that provides actionable information about $\mathrm{Weyl}(\rho)$ is obtained, the computational difference sampling step does not depend on the choice of the Clifford ensemble. However, using shallow-depth circuits reduces measurement noise, which can improve the efficiency of classical post-processing.

In Fig.~\ref{fig:2}(b), we compare the recovery of hidden symmetries under noise using computational difference sampling. To reconstruct the hidden bitstring from noisy samples, we use the Fast Walsh--Hadamard Transform (FWHT) correlation test, with details provided in the Appendix~\ref{appx:numeric}.

In the experiment, we vary the number of samples $N_S$ and evaluate the success
probability $p_{\mathrm{suc}}$ of correctly recovering the hidden symmetry for a
system with $n=24$ and $t=0$. Noise is introduced by flipping each bit
independently with probability $p_{\mathrm{err}}$. Each data point is averaged over $1000$
trials. As shown in the figure, the number of samples required for successful
post-processing increases with the noise level. The common initial value across all $p_{\mathrm{err}}$ arises because, when $\mathrm{Weyl}(\rho) \cap C^{\dagger}(\mathcal{Z}) = \{0\}$, the correlation values are negligible regardless of the noise level.

These numerical simulations suggest that, in the presence of noise, using shallow
unitary ensembles not only simplifies circuit implementation but also reduces the
number of experimental repetitions required for reliable post-processing.

\section{Lower bound}
Known single-copy algorithms for learning an $(n-t)$-dimensional stabilizer group
require a number of samples that scales exponentially in $t$~\cite{grewal2025efficient, chia2024efficient}. It is therefore natural
to ask whether this dependence is merely an artifact of existing techniques, or whether
it is in fact inherent to the single-copy setting. The following theorem shows that the
latter is the case: in the worst case, any algorithm based on single-copy measurements
must use at least $\Omega(2^t)$ samples. Consequently, when $t=\omega(\log n)$, both
the sample complexity and the overall runtime become super-polynomial.

\begin{theorem}[Single-copy lower bound on $t$-dependence]
    Any algorithm that uses potentially adaptive single-copy measurements for
    $(\epsilon,\delta,t)$-stabilizer group learning must, in the worst case, use at least
    $\Omega(2^t)$ samples.
\end{theorem}

A detailed proof is given in the Appendix~\ref{appx:lower}, where we reduce our stabilizer group learning problem to a previously studied single-copy Pauli testing problem and invoke the known lower bound~\cite{huang2021information, chen2022exponential, chen2024optimal}.

Combining this hardness result for classical learners with the sample-efficient Bell difference sampling protocol for quantum learners~\cite{grewal2025efficient, hangleiter2024bell, chen2025stabilizer, leone2024learning, gross2021schur} yields a quantum advantage for learning Pauli symmetries of quantum systems. This separation is notable for several reasons. First, once a candidate symmetry group is found, it can be verified efficiently. Second, the quantum algorithm is not only sample-efficient but also computationally efficient, in contrast to previously known examples where the advantage appears only at the level of sample complexity~\cite{huang2021information}. Finally, identifying symmetries of quantum states is a natural problem in many-body physics, giving this learning problem direct physical motivation and potential applications.

\section{Discussion and Outlook}
We studied the problem of learning stabilizer groups using single-copy measurements. On the positive side, we showed that shallow-depth measurements suffice to efficiently learn almost all stabilizer groups, using a refined counting argument that extends from Lagrangian to general isotropic subspaces in symplectic vector spaces. On the negative side, we established a limitation of arbitrary single-copy measurements by showing that, in the worst case, learning can require exponentially more samples than in the two-copy setting~\cite{grewal2025efficient, hangleiter2024bell, chen2025stabilizer} when $t$ is large.

A natural direction for future work is to understand whether, by allowing adaptivity, one can learn stabilizer groups using only shallow-depth measurements even in the worst case. It is also interesting to study what can be proved when the measurements are further restricted. As we observed in numerical experiments and previous work~\cite{arunachalam2022optimal}, even single-qubit measurements are already very effective. 
Another interesting direction is to consider variants of the block-Clifford ensemble used in this work. If the block-Clifford ensemble is arranged in a brickwork pattern, one obtains a two-layer block-Clifford ensemble, which can be used to construct approximate unitary designs~\cite{schuster2025random}. It would be interesting to understand whether the two-layer block-Clifford ensemble is sufficient in the worst case.

\begin{acknowledgments}
This work was supported by a National Research Foundation of Korea (NRF) grant funded
by the Korean Government (MSIT) (No. 2019M3E4A1080144, No. 2019M3E4A1080145, No.
2019R1A5A1027055, RS-2023-00283291, SRC Center for Quantum Coherence in Condensed
Matter RS-2023-00207732, quantum computing technology development program No.
2020M3H3A1110365, and No. 2023R1A2C2005809, No. RS-2024-00413957, No. RS-2024-
00442994), a Korea Basic Science Institute (National Research Facilities and Equipment
Center) grant funded by the Ministry of Education (No. 2021R1A6C101B418).
\end{acknowledgments}

\bibliography{main_ref}
\appendix
\onecolumngrid

\section{Related work}
\paragraph{Multi-copy and single-copy stabilizer learning.}
Previous work on stabilizer learning has explored both multi-copy~\cite{grewal2025efficient, leone2024learning, hangleiter2024bell, montanaro2017learning} and single-copy
measurement settings~\cite{grewal2025efficient, chia2024efficient}. Multi-copy methods are often more efficient, but typically rely
on entangled measurements across different copies of the state and may require coherent
quantum memory. Single-copy approaches avoid this requirement and are therefore more
appealing from an experimental point of view, but existing methods usually rely on
random Clifford measurements implemented by circuits of depth $O(n)$. In contrast, our
result shows that, for almost all stabilizer groups $\mathrm{Weyl}(\rho)$ of dimension
$n-t$ with $t=O(\log n)$, one can learn the hidden Pauli symmetries without coherent
quantum memory and without using linear-depth measurement circuits.

\paragraph{State hidden subgroup problem.}
Our learning problem can be viewed as a special case of the recently introduced state
hidden subgroup problem~\cite{hinsche2025abelian, bouland2025state}, which asks for the hidden symmetry subgroup of an unknown
quantum state. From this perspective, learning stabilizer groups is one representative
example, alongside other tasks such as identifying cuts across which the state is
unentangled and learning hidden translation symmetries. The common theme is that one
aims to recover a symmetry of the state rather than a complete classical description
of it. The standard algorithmic approach to this framework is based on Fourier sampling~\cite{bouland2025state}.
However, implementing it in full generality typically relies on generalized phase
estimation, which requires auxiliary registers and controlled group actions and can
therefore be demanding in practice. More recent work~\cite{hinsche2025abelian} has shown that, in several
important cases, the same structure can be realized through simpler measurement primitives, including two-copy routines such as Bell sampling and Bell difference sampling.

\paragraph{Stabilizer testing.}
A related line of work considers stabilizer testing~\cite{hinsche2025single, gross2021schur} rather than stabilizer learning. Here the
task is only to determine whether an unknown state is a stabilizer state or is $\epsilon$-far from the
stabilizer family, so the problem is generally easier than full learning. Nevertheless, as in
learning, the measurement model plays a decisive role. Two-copy protocols based on Bell
difference sampling suffice with \(O(1)\) copies~\cite{gross2021schur}, while single-copy protocols are inherently
more costly. Recent work~\cite{hinsche2025single} showed that single-copy stabilizer testing can be carried out
using computational difference sampling with \(O(n)\) copies, and also proved a lower bound
of \(\Omega(\sqrt{n})\) copies for any single-copy algorithm. This again highlights the gap
between single-copy and multi-copy access.

\paragraph{Randomized measurements and quantum state tomography.}
The block-Clifford ensemble $\mathcal{C}_k=\mathrm{Cl}(k)^{\otimes n/k}$ used in our work
has also appeared in recent studies on randomized measurements~\cite{cho2025entanglement} and quantum state tomography~\cite{cho2025sample}.
In the randomized-measurement setting, such ensembles provide a hardware-efficient way to
estimate a variety of quantities, ranging from linear observables such as $\mathrm{tr}(O\rho)$ to
nonlinear ones such as the purity $\mathrm{tr}(\rho^2)$. In this line of work, the choice of unitary ensemble affects not only the sample complexity,
but also the cost of the classical post-processing. While earlier approaches~\cite{bertoni2024shallow, schuster2025random, hu2025demonstration} often
prioritized sample efficiency, they could require more delicate or approximate post-processing. Recent work~\cite{cho2025entanglement} showed that the block-Clifford ensemble
$\mathcal{C}_k=\mathrm{Cl}(k)^{\otimes n/k}$ admits efficient classical post-processing while maintaining optimal sample complexity guarantees in many cases. This efficient post-processing also makes it possible to extend hybrid quantum-classical learning schemes, previously developed for single-qubit measurements~\cite{huang2022provably}, to the setting of shallow-depth randomized measurements.

The same Clifford ensemble $\mathcal{C}_k$ has also been used for quantum state tomography~\cite{cho2025sample}. Previous optimal single-copy tomography protocols typically relied on linear-depth Clifford measurements~\cite{lowe2025lower}. By contrast, recent work showed that $\mathcal{C}_k$ with $k=O(\log n)$ already suffices to achieve optimal sample complexity in the full-rank case, thereby reducing the required measurement depth to $O(\log n)$ without sacrificing performance guarantees.

\section{Weyl representation and symplectic geometry}

In the main text, we ignore global phases of Pauli operators and identify unsigned
$n$-qubit Pauli operators with binary vectors in
\[
V := \mathbb{F}_2^{2n}
    = \mathbb{F}_2^n \times \mathbb{F}_2^n.
\]
For $a=(a_x,a_z)\in V$, we define the corresponding
Weyl operator by
\begin{equation}
    W_a
    :=
    i^{a_x\cdot a_z}
    \bigotimes_{j=1}^n
    X^{(a_x)_j} Z^{(a_z)_j},
\end{equation}
where the exponent $a_x\cdot a_z$ in the phase factor is evaluated over the integers.
Modulo global phase, every $n$-qubit Pauli operator is represented uniquely by a vector
$a\in V$. Thus, multiplication of unsigned Pauli operators corresponds to addition in
$\mathbb{F}_2^{2n}$.

The vector space $V$ is equipped with the canonical symplectic bilinear form
\begin{equation}
    [a,b]
    :=
    a_x \cdot b_z + a_z \cdot b_x
    \pmod 2,
    \qquad
    a=(a_x,a_z),\; b=(b_x,b_z)\in V,
\end{equation}
where $\cdot$ denotes the standard inner product on $\mathbb{F}_2^n$.
This form encodes the commutation relations of Weyl operators via
\begin{equation}
    W_a W_b = (-1)^{[a,b]} W_b W_a .
\end{equation}
In particular, $W_a$ and $W_b$ commute if and only if $[a,b]=0$.
Thus, commutation among Pauli operators is naturally captured by the symplectic geometry of $\mathbb{F}_2^{2n}$.

For $A\subseteq V$, we define its symplectic complement by
\begin{equation}
    A^\perp
    :=
    \{x\in V : [x,a]=0 \text{ for all } a\in A\}.
\end{equation}
A subset $A\subseteq V$ is called \emph{isotropic} if $[x,y]=0$ for all $x,y\in A$, that is, if all Weyl operators $\{W_x:x\in A\}$ commute pairwise.
An isotropic subspace is called \emph{Lagrangian} if it has maximal possible dimension, namely $\dim(A)=n$.

We will use the following standard facts for subspaces $A,B \subseteq V$:
\begin{itemize}
    \item $(A^\perp)^\perp = A$.
    \item $\dim(A)+\dim(A^\perp)=2n$; equivalently, $|A|\,|A^\perp|=4^n$.
    \item $A\subseteq B$ if and only if $B^\perp \subseteq A^\perp$.
    \item $(A+B)^\perp = A^\perp \cap B^\perp$.
\end{itemize}
Moreover, if $A$ is isotropic, then $A\subseteq A^\perp$, and therefore $\dim(A)\le n$.
\begin{fact}\label{fact:sum_perp}
Let $A,B \subset V$ be subspaces. Then
\begin{equation}
    (A+B)^\perp = A^\perp \cap B^\perp.
\end{equation}
\end{fact}

\begin{fact}\label{fact:modular_identity}
Let $A,B,C \subset V$ be subspaces such that $C \subseteq A$. Then
\begin{equation}
    A \cap (B+C) = (A\cap B)+C.
\end{equation}
\end{fact}
We denote by $\mathrm{Sp}(V)$ the symplectic group of $V$, namely
\begin{equation}
    \mathrm{Sp}(V)
    :=
    \{\phi \in \mathrm{GL}(V) : [\phi(x),\phi(y)] = [x,y] \text{ for all } x,y \in V\},
\end{equation}
where $\mathrm{GL}(V)$ is the group of invertible linear maps on $V$.

For two $\mathbb{F}_2$-vector spaces $A$ and $B$, we denote by
\begin{equation}
    \mathrm{Hom}(A,B)
\end{equation}
the space of all $\mathbb{F}_2$-linear maps from $A$ to $B$.
In particular, if $\dim A = a$ and $\dim B = b$, then $\mathrm{Hom}(A,B)$ is an
$\mathbb{F}_2$-vector space of dimension $ab$, and hence
\begin{equation}
    |\mathrm{Hom}(A,B)| = 2^{ab}.
\end{equation}

In particular, a subspace $L \subset V$ is Lagrangian when it is isotropic and has dimension $n$.
Equivalently, if $L\subseteq V$ is Lagrangian, then $L=L^\perp$.

We will frequently use the coordinate subspaces
\begin{equation}
    \mathcal{X}:=\{(x,0^n):x\in\mathbb{F}_2^n\},
    \qquad
    \mathcal{Z}:=\{(0^n,z):z\in\mathbb{F}_2^n\},
\end{equation}
both of which are Lagrangian subspaces of $V$.

We will also use the following standard character-sum identities.
For a subspace $A\subseteq \mathbb{F}_2^n$, define its orthogonal complement with respect
to the standard inner product by
\begin{equation}
    A^{\perp_s}
    :=
    \{b\in \mathbb{F}_2^n : a\cdot b=0 \text{ for all } a\in A\}.
\end{equation}
Then, for every $b\in \mathbb{F}_2^n$,
\begin{equation}
    \sum_{a\in A}(-1)^{a\cdot b}
    =
    |A|\,\mathbf{1}\{b\in A^{\perp_s}\}.
\end{equation}
Similarly, for a subspace $A\subseteq V=\mathbb{F}_2^{2n}$ and every $b\in V$,
\begin{equation}
    \sum_{a\in A}(-1)^{[a,b]}
    =
    |A|\,\mathbf{1}\{b\in A^\perp\}.
\end{equation}

For proofs in the following sections, we repeatedly use the quotient map
\begin{equation}
    \pi : M^{\perp} \to M^{\perp}/M, \qquad \pi(x)=x+M,
\end{equation}
where \(M\subset V\) is isotropic. For any Lagrangian subspace \(L\subset V\), the image
\begin{equation}
    \pi(L\cap M^{\perp})=(L\cap M^{\perp}+M)/M
\end{equation}
is a Lagrangian subspace of \(M^{\perp}/M\), since it is isotropic and has dimension \(n-\dim M\).

The Weyl operators form an orthogonal basis of the operator space on $(\mathbb{C}^2)^{\otimes n}$
with respect to the Hilbert--Schmidt inner product. More explicitly,
\begin{equation}
    \mathrm{tr}(W_a W_b)=2^n \delta_{a,b}.
\end{equation}
Therefore any $n$-qubit state $\rho$ admits the Weyl expansion
\begin{equation}
    \rho
    =
    \frac{1}{2^n}
    \sum_{a\in V}
    \mathrm{tr}(\rho W_a)\, W_a .
\end{equation}
This expansion shows that the coefficients $\mathrm{tr}(\rho W_a)$ provide a complete description
of the state in the Weyl basis.

Finally, following the notation in the main text, we define
\begin{equation}
    \mathrm{Weyl}(\rho)
    :=
    \{a\in V : \mathrm{tr}(\rho W_a)^2 = 1\}.
\end{equation}
Thus, $\mathrm{Weyl}(\rho)$ records the unsigned Pauli symmetries of $\rho$.
Because $[a,b]=0$ is equivalent to commutation of $W_a$ and $W_b$, the set
$\mathrm{Weyl}(\rho)$ is naturally viewed as an isotropic subspace of the
vector space $V$.

We will also use the Gaussian binomial coefficient. For integers $0 \le r \le m$ and a prime power $q$, define
\begin{equation}
    \binom{m}{r}_q
    :=
    \prod_{i=0}^{r-1}\frac{q^{m-i}-1}{q^{r-i}-1}.
\end{equation}
Equivalently, $\binom{m}{r}_q$ is the number of $r$-dimensional subspaces of $\mathbb{F}_q^m$.
In particular, when $q=2$, $\binom{m}{r}_2$ counts the number of $r$-dimensional subspaces of $\mathbb{F}_2^m$. We will use the following standard properties of the Gaussian binomial coefficients:
\begin{flalign}
    &\quad1. \quad \binom{m}{r}_2=\binom{m}{m-r}_2\label{prop:gaissian_binom_one},\\
    &\quad2. \quad \prod_{i=0}^{m-1}(1+2^i t)
    =
    \sum_{i=0}^{m}2^{\binom{i}{2}}\binom{m}{i}_2\,t^i.&\label{prop:gaissian_binom}
\end{flalign}

\begin{fact}\label{fact:num_largrangian}
    The number of Lagrangian subspaces of $V=\mathbb{F}_2^{2n}$ is $\prod_{i=1}^{n}(2^i+1)$.
\end{fact}

\begin{lemma}[The number of isotropic subspaces]\label{lem:number-of-isotropic-subspaces}
Let $m \in [n]$, and let $\mathcal{I}_{m}(V)$ denote the set of isotropic subspaces of
$V=\mathbb{F}_2^{2n}$ of dimension $m$.
Then
\begin{equation}
    |\mathcal{I}_{m}(V)|
    =
    \binom{n}{m}_2
    \prod_{i=n-m+1}^{n}(2^{i}+1).
\end{equation}
\end{lemma}

\begin{proof}
    We count ordered bases $(v_1,\dots,v_m)$ of a $m$-dimensional isotropic subspace. The first vector $v_1$ can be any nonzero vector of $V$, so there are $2^{2n}-1$ choices.
    Suppose inductively that $v_1,\dots,v_i$ have already been chosen and span an isotropic subspace
    $U_i$ of dimension $i$. Since $U_i$ is isotropic, we have $U_i \subseteq U_i^\perp$, and $\dim(U_i^\perp)=2n-i$. To keep the span isotropic, the next vector $v_{i+1}$ must lie in $U_i^\perp \setminus U_i$.
    Hence, the number of choices for $v_{i+1}$ is $|U_i^\perp \setminus U_i|=2^{2n-i}-2^i$. Therefore, the total number of ordered linearly independent $m$-tuples spanning an isotropic subspace is
    \begin{equation}
        \prod_{i=0}^{m-1}(2^{2n-i}-2^i).
    \end{equation}
    Each $m$-dimensional isotropic subspace has 
    \begin{equation}
        \prod_{i=0}^{m-1}(2^m-2^i)
    \end{equation}
    ordered bases. Dividing, we obtain
    \begin{align}
        |\mathcal{I}_{m}(V)|
        &=
        \frac{\prod_{i=0}^{m-1}(2^{2n-i}-2^i)}
             {\prod_{i=0}^{m-1}(2^m-2^i)}\\
        &=
        \prod_{i=0}^{m-1}
        \frac{2^{2n-2i}-1}{2^{m-i}-1}\\
        &=
        \prod_{i=0}^{m-1}
        \frac{(2^{n-i}-1)(2^{n-i}+1)}{2^{m-i}-1} \\
        &=
        \left(
        \prod_{i=0}^{m-1}\frac{2^{n-i}-1}{2^{m-i}-1}
        \right)
        \left(
        \prod_{i=0}^{m-1}(2^{n-i}+1)
        \right) \\
        &=
        \binom{n}{m}_2
        \prod_{i=0}^{m-1}(2^{n-i}+1)\\
        &=
        \binom{n}{m}_2
        \prod_{i=n-m+1}^{n}(2^{i}+1).
    \end{align}
    For $m=n$, this recovers Fact~\ref{fact:num_largrangian}.
\end{proof}

\section{Computational difference sampling}

Computational difference sampling~\cite{grewal2025efficient} uses two independent copies of $\rho$. Measuring each copy in the computational basis gives outcomes $b_1,b_2\in \mathbb{F}_2^n$, and the procedure outputs their bitwise mod-$2$ sum $b_1+b_2\in \mathbb{F}_2^n$. In the main text, we denote the resulting distribution by $r_\rho$. Unlike Bell difference sampling, the two measurements are performed independently, so this procedure can be implemented in the single-copy setting and requires no quantum memory or ancillary qubits.

The natural domain of $r_\rho$ is $\mathbb{F}_2^n$. However, to match the Weyl representation in the sympectic geometry, it is convenient to regard $r_\rho$ as a distribution on $    \mathcal{X}\coloneqq \{(x,0^n): x\in \mathbb{F}_2^n\}\subseteq \mathbb{F}_2^{2n}$. Then, for $a\in \mathbb{F}_2^n$,
\begin{equation}
    r_{\rho}(a,0^n)
    =
    \sum_{x \in \mathbb{F}_2^n}
    \mathrm{tr}\bigl(\rho \ket{x}\bra{x}\bigr)\,
    \mathrm{tr}\bigl(\rho \ket{x+a}\bra{x+a}\bigr).
\end{equation}

The next lemma rewrites $r_\rho$ in terms of Pauli expectation values.

\begin{lemma}\label{lem:cds-pauli-expansion}
    For every $a\in \mathbb{F}_2^n$,
    \begin{equation}
        r_\rho(a,0^n)
        =
        2^{-n}\sum_{b\in \mathbb{F}_2^n}
        (-1)^{a\cdot b}\,\mathrm{tr}(\rho Z^b)^2.
    \end{equation}
\end{lemma}
\begin{proof}
    Using $X^a\ket{x}=\ket{x+a}$, we have
    \begin{align}
        r_\rho(a,0^n)
        &=
        \sum_{x\in \mathbb{F}_2^n}
        \mathrm{tr}(\rho \ket{x}\bra{x})\,
        \mathrm{tr}(\rho X^a\ket{x}\bra{x}X^a) \\
        &=
        2^{-n}\sum_{b\in \mathbb{F}_2^n}
        \mathrm{tr}(\rho Z^b)\,
        \mathrm{tr}(\rho X^a Z^b X^a) \\
        &=
        2^{-n}\sum_{b\in \mathbb{F}_2^n}
        (-1)^{a\cdot b}\,\mathrm{tr}(\rho Z^b)^2,
    \end{align}
    where $X^a=\bigotimes_i X^{a_i}$, $Z^b=\bigotimes_i Z^{b_i}$, and $a\cdot b=\sum_i a_i b_i$. In the second line, we used
    \[
        \sum_{x\in \mathbb{F}_2^n}\ket{x}\bra{x}\otimes\ket{x}\bra{x}
        =
        2^{-n}\sum_{b\in \mathbb{F}_2^n} Z^b\otimes Z^b.
    \]
\end{proof}

For $A\subseteq \mathcal{X}$, write $r_\rho(A)\coloneqq \sum_{x\in A} r_\rho(x)$. We next compute the mass of an $X$-type subspace under $r_\rho$.

\begin{lemma}\label{lem:cds-subspace-mass}
    Let $H=H_X\times 0^n\subseteq \mathcal{X}$, where $H_X\subseteq \mathbb{F}_2^n$ is a subspace. Then
    \begin{equation}
        r_\rho(H)
        =
        \frac{|H+\mathcal{Z}|}{2^n}\,p_\rho((H+\mathcal{Z})^\perp),
    \end{equation}
    where $p_\rho(S)\coloneqq \sum_{x\in S} p_\rho(x)$.
\end{lemma}

\begin{proof}
    By Lemma~\ref{lem:cds-pauli-expansion},
    \begin{align}
        r_\rho(H)
        &=
        2^{-n}\sum_{b\in \mathbb{F}_2^n}\sum_{a\in H_X}
        (-1)^{a\cdot b}\,\mathrm{tr}(\rho Z^b)^2 \\
        &=
        2^{-n}|H_X|
        \sum_{b\in H_X^{\perp_s}}\mathrm{tr}(\rho Z^b)^2 \\
        &=
        |H_X|\,p_\rho(0^n\times H_X^{\perp_s})\\
        &=\frac{|H+\mathcal{Z}|}{2^n}\,p_\rho((H+\mathcal{Z})^\perp),
    \end{align}
    where $H_X^{\perp_s}\coloneqq \{b\in \mathbb{F}_2^n : a\cdot b=0 \text{ for all } a\in H_X\}$. Since $H+\mathcal{Z}=H_X\times \mathbb{F}_2^n$, we have $|H+\mathcal{Z}|=2^n|H_X|$, and since $(H+\mathcal{Z})^\perp=0^n\times H_X^{\perp_s}$, the last equality follows.
\end{proof}

The next lemma identifies the orthogonal complement that will appear in the support statement.

\begin{lemma}\label{lem:cds-orthogonal}
    Let $H\coloneqq (\mathrm{Weyl}(\rho)\cap \mathcal{Z})^\perp \cap \mathcal{X}$. Then
    \begin{equation}
        (H+\mathcal{Z})^\perp
        =
        \mathrm{Weyl}(\rho)\cap \mathcal{Z}.
    \end{equation}
\end{lemma}

\begin{proof}
    Using $(A+B)^\perp=A^\perp\cap B^\perp$, we get
    \begin{align}
        (H+\mathcal{Z})^\perp
        &=
        H^\perp\cap \mathcal{Z} \\
        &=
        (\mathrm{Weyl}(\rho)\cap \mathcal{Z}+\mathcal{X})\cap \mathcal{Z} \\
        &=
        \mathrm{Weyl}(\rho)\cap \mathcal{Z},
    \end{align}
    where we used $\mathcal{X}\cap \mathcal{Z}=\{0\}$ in the last step.
\end{proof}

\begin{corollary}[Support of computational difference sampling]\label{prop:cds-support}
    If $r_\rho(x)>0$, then
    \begin{equation}
        x\in (\mathrm{Weyl}(\rho)\cap \mathcal{Z})^\perp \cap \mathcal{X}.
    \end{equation}
\end{corollary}
\begin{proof}
    Let $H\coloneqq (\mathrm{Weyl}(\rho)\cap \mathcal{Z})^\perp \cap \mathcal{X}$. By Lemma~\ref{lem:cds-subspace-mass},
    \[
        r_\rho(H)
        =
        \frac{|H+\mathcal{Z}|}{2^n}\,p_\rho((H+\mathcal{Z})^\perp).
    \]
    By Lemma~\ref{lem:cds-orthogonal}, $(H+\mathcal{Z})^\perp=\mathrm{Weyl}(\rho)\cap \mathcal{Z}\subset\mathrm{Weyl}(\rho)$. Hence
    \[
        r_\rho(H)
        =
        \frac{|H+\mathcal{Z}|\cdot |(H+\mathcal{Z})^\perp|}{4^n}
        =
        1.
    \]
    Therefore, the support of $r_\rho$ is contained in $H$.
\end{proof}

\section{Clifford ensembles}\label{appx:clifford ensemble}

In the main text, instead of the full random Clifford ensemble $\mathrm{Cl}(n)$, we consider the block-Clifford ensemble
\begin{equation}
    \mathcal{C}_k \coloneqq \mathrm{Cl}(k)^{\otimes n/k},
\end{equation}
where we assume throughout the paper that $k$ divides $n$. Here, each block has size $k$, and the Clifford circuit acting on each block is sampled independently from $\mathrm{Cl}(k)$. The ensemble $\mathcal{C}_k$ is useful for randomized measurements and quantum state tomography because of its tensor-product structure and because it forms a group.

For a given unitary ensemble $\mathcal{C}$, an important quantity in our analysis is the Pauli collision probability
\begin{equation}
    m_P \coloneqq \mathbb{E}_{C \sim \mathcal{C}} \left[\mathbf{1}\{ C(P) \in \mathcal{Z} \}\right],
\end{equation}
where $\mathcal{Z} = \{I,Z\}^n=\{(0^n, z):z\in\mathbb{F}_2^{n}\}$. For $\mathcal{C}_k$, this quantity is given by
\begin{align}
    m_P
    &= \mathbb{E}_{C \sim \mathcal{C}_k} \left[\mathbf{1}\{ C(P) \in \mathcal{Z} \}\right] \\
    &= \prod_{i=1}^{n/k} \mathbb{E}_{C_i \sim \mathrm{Cl}(k)}
    \left[\mathbf{1}\{ C_i (P_i) \in \mathcal{Z}_k \}\right] \\
    &= (2^k+1)^{-w_k(P)},
\end{align}
where $C = \bigotimes_{i=1}^{n/k} C_i$, $P = \bigotimes_{i=1}^{n/k} P_i$, and $\mathcal{Z}_k = \{I,Z\}^k=\{(0^k, z):\in\mathbb{F}_2^{k}\}$. Here,
\begin{equation}
    w_k(P) \coloneqq \left| \{ i \in [n/k] : P_i \neq I_k \} \right|
\end{equation}
denotes the number of blocks on which $P$ acts nontrivially, and $I_k$ is the $2^k \times 2^k$ identity matrix.

Since $w_k(P) \le n/k$ for all $P$, we obtain
\begin{equation}
    \min_P m_P = (2^k+1)^{-n/k}
    = 2^{-n}(1+2^{-k})^{-n/k},
\end{equation}
where equality is attained when $P$ is nontrivial on every block. Now let $f(x) = (1+x^{-1})^x$ for $x>0$. Since $f(x)$ is monotone increasing and $\lim_{x\to\infty} f(x)=e$, it follows that
\begin{equation}
    e^{-n/(2^k k)} 2^{-n} < \min_P m_P < 2^{-n}.
\end{equation}
Therefore, if the block size $k$ is chosen so that $2^k k = \Omega(n)$, then $e^{-n/(2^k k)} = \Omega(1)$. In particular, when $k = \Omega(\log n)$, we have
\begin{equation}\label{eq:mp_lower}
    \min_P m_P = \Theta(2^{-n}).
\end{equation}

Since implementing a unitary in $\mathrm{Cl}(k)$ requires depth $O(k)$~\cite{MaslovRoetteler2018}, choosing $k = O(\log n)$ allows us to realize $\mathcal{C}_k$ with circuit depth $O(\log n)$.

\section{Performance guarantees of the learning algorithm}\label{appx:proof}
In this section, we give a detailed proof of Theorem~\ref{thm1} from the main text. We first recall the targeted learning problem.

\begin{definition}[{$(\epsilon,\delta,t)$-stabilizer group learning}]
    Let $\rho$ be an $n$-qubit state with stabilizer dimension $n-t$. An algorithm outputs an isotropic subspace $\widehat S$ such that, with probability at least $1-\delta$,
    \begin{enumerate}
        \item $\widehat S \supseteq \mathrm{Weyl}(\rho)$,
        \item $\mathbb{E}_{x\sim \widehat S}\!\left[\mathrm{tr}(\rho W_x)^2\right] \ge 1-\epsilon$.
    \end{enumerate}
\end{definition}

We analyze Algorithm~\ref{alg}, which learns $\mathrm{Weyl}(\rho)$ using only single-copy shallow-depth measurements. The proof has two parts: first, we show that random Clifford circuits from $\mathcal{C}_k$ reveal enough information to recover $\mathrm{Weyl}(\rho)$; second, we bound the number of computational difference samples needed to recover the hidden Pauli symmetries in each selected basis.

Before proving the theorem, we first collect several results about averages over uniformly random isotropic subspaces $S\sim \mathcal{I}_{n-t}(V)$ that will be used below.

\begin{fact}[Theorem 4 in~\cite{kueng2015qubit}]\label{fact:num_lag_intersection}
    Let $L_0 \subset V\coloneq\mathbb{F}_2^{2n}$ be a fixed Lagrangian subspace. Then, for each $0 \le l \le n$, the number of Lagrangian subspaces $L \subset V$ such that $\dim(L_0 \cap L)=l$ is
    \begin{equation}
        \left|\left\{L \subset V : \dim(L \cap L_0)=l \right\}\right|
        = \binom{n}{l}_2 \, 2^{(n-l)(n-l+1)/2}.
    \end{equation}
    In particular, the total number of Lagrangian subspaces of $V$ is
    \begin{equation}
        \prod_{i=1}^{n}(2^i+1).
    \end{equation}
\end{fact}

\begin{definition}\label{def:generating_func}
    Let $L_0 \subset V\coloneq\mathbb{F}_2^{2n}$ be a fixed Lagrangian subspace, and let
    $L$ be sampled uniformly from the set of all Lagrangian subspaces $\mathcal{I}_n(V)$.
    We define the generating function by
    \begin{align}
        G_n(s)
        &\coloneq \sum_{l=0}^{n} \Pr_L\bigl(\dim(L\cap L_0)=l\bigr) s^l \\
        &= \sum_{l=0}^{n}
        \frac{\binom{n}{l}_2\,2^{(n-l)(n-l+1)/2}}{\prod_{i=1}^{n}(2^i+1)}\, s^l \\
        &= \frac{s^n}{\prod_{i=1}^n (2^{i}+1)}
        \prod_{i=0}^{n-1}\left(1+\frac{2^{i+1}}{s}\right),
    \end{align}
    where the second equality uses Fact~\ref{fact:num_lag_intersection}, and the last equality follows from Eq.~\eqref{prop:gaissian_binom}.
    In particular, we will use the special values
    \begin{align}
        G_n(2) &= \frac{2^{n+1}}{2^n+1}, \\
        G_n(4) &= \frac{6\cdot 4^n}{(2^n+1)(2^n+2)}.
    \end{align}
\end{definition}

\begin{lemma}\label{lem:zero_intersec}
    Let $L \subset V\coloneq\mathbb{F}_2^{2n}$ be a fixed Lagrangian subspace. Then, for every
    $0 \le m \le n$,
    \begin{equation}
        \bigl|\{S \in \mathcal{I}_m(V) : \dim(S \cap L)=0\}\bigr|
        =
        \binom{n}{m}_2\,
        2^{m(m+1)/2}\,2^{m(n-m)}.
    \end{equation}
\end{lemma}
\begin{proof}
    Fix a Lagrangian complement $M$ of $L$, so that
    \begin{equation}
        V=L\oplus M.
    \end{equation}
    Write each $x\in V$ uniquely as $x=x_L+x_M$ with $x_L\in L$ and $x_M\in M$. Let $S\in \mathcal I_m(V)$ satisfy $S\cap L=\{0\}$ and let $\pi_M:V\to M, \pi(x)=x_M$. Then $\pi_M|_S:S\to M$ is injective,
    since $\ker(\pi_M|_S)=S\cap L=\{0\}$. Thus
    \begin{equation}
        A\coloneq \pi_M(S)
    \end{equation}
    is an $m$-dimensional subspace of $M$. Moreover, $\pi_M|_S:S\to A$ is an isomorphism, so
    there is a unique linear map $\phi:A\to L$ such that
    \begin{equation}
        S=\Gamma(\phi)\coloneq \{a+\phi(a):a\in A\}.
    \end{equation}
    Conversely, every pair $(A,\phi)$ with $A\subseteq M$, $\dim A=m$, and
    $\phi\in \mathrm{Hom}(A,L)$ gives such a subspace. Hence it suffices to count those
    $\phi$ for which $\Gamma(\phi)$ is isotropic.
    
    For $a,b\in A$, since $A\subset M$ and $\phi(A)\subset L$, and both $L$ and $M$ are
    isotropic, we have
    \begin{equation}
        [a+\phi(a),\,b+\phi(b)]
        =[a,\phi(b)]+[\phi(a),b].
    \end{equation}
    Thus $\Gamma(\phi)$ is isotropic if and only if
    \begin{equation}
        [a,\phi(b)]=[b,\phi(a)]
        \qquad\text{for all }a,b\in A.
    \end{equation}
    
    Now $A^\perp\cap L$ is a subspace of $L$ of dimension $n-m$. Choose a complementary
    subspace $B\subset L$ such that
    \begin{equation}
        L=B\oplus (A^\perp\cap L),
    \end{equation}
    so $\dim B=m$. Then every $\phi\in \mathrm{Hom}(A,L)$ decomposes uniquely as
    \begin{equation}
        \phi=\phi_B+\phi_{\mathrm{free}},
    \end{equation}
    where $\phi_B\in \mathrm{Hom}(A,B)$ and
    $\phi_{\mathrm{free}}\in \mathrm{Hom}(A,A^\perp\cap L)$. Since
    \begin{equation}
        [A,A^\perp\cap L]=0,
    \end{equation}
    the map $\phi_{\mathrm{free}}$ does not affect the isotropy condition, so only $\phi_B$ matters.
    
    Fix bases of $A$ and $B$. Then $\phi_B$ is represented by an $m\times m$ matrix, and
    the condition
    \begin{equation}
        [a,\phi_B(b)]=[b,\phi_B(a)]
    \end{equation}
    means exactly that this matrix is symmetric. Hence the number of possible $\phi_B$ is the number of symmetric
    $m\times m$ matrices over $\mathbb{F}_2$, namely
    \begin{equation}
        2^{m(m+1)/2}.
    \end{equation}
    Also, $\phi_{\mathrm{free}}$ can be chosen arbitrarily in
    $\mathrm{Hom}(A, A^\perp \cap L)$. Therefore, the number of possible
    choices for $\phi_{\mathrm{free}}$ is
    \begin{equation}
        |\mathrm{Hom}(A, A^\perp \cap L)| = 2^{m(n-m)}.
    \end{equation}
    Therefore, for each fixed $A\subseteq M$ with $\dim A=m$, the number of isotropic subspaces
    $S$ with $\pi_M(S)=A$ and $S\cap L=\{0\}$ is
    \begin{equation}
        2^{m(m+1)/2}2^{m(n-m)}.
    \end{equation}
    
    Finally, the number of $m$-dimensional subspaces $A\subseteq M$ is $\binom{n}{m}_2$, so
    \begin{equation}
        \bigl|\{S\in\mathcal I_m(V):\dim(S\cap L)=0\}\bigr|
        =
        \binom{n}{m}_2\,2^{m(m+1)/2}\,2^{m(n-m)},
    \end{equation}
    as desired.
\end{proof}

Extending the Lagrangian counting result in Fact~\ref{fact:num_lag_intersection} to isotropic subspaces $S \in \mathcal{I}_m(V)$, we obtain the following.

\begin{lemma}\label{lem:iso_lag_intersection}
    Let $L \subset V\coloneq\mathbb{F}_2^{2n}$ be a fixed Lagrangian subspace. Then, for every
    $0 \le l \le m \le n$,
    \begin{equation}
        \bigl|\{S \in \mathcal{I}_m(V) : \dim(S \cap L)=l\}\bigr|
        =
        \binom{n}{l}_2 \binom{n-l}{m-l}_2\,
        2^{(m-l)(m-l+1)/2}\,2^{(m-l)(n-m)}.
    \end{equation}
\end{lemma}
\begin{proof}  
    Fix $T=S\cap L$ with $\dim(T)=l$. There are $\binom{n}{l}_2$ such subspaces $T\subset L$. For each fixed such $T$, we have
    \begin{equation}
        S \cap L=T \quad\Leftrightarrow\quad S_T\cap L_{T}=\{0\},
    \end{equation}
    where $S_T=S/T$ and $L_T=L/T$. Therefore, for fixed $T$,
    \begin{equation}
        |\{S\in\mathcal{I}_{m}(V):S\cap L=T, \ \dim(T)=l\}|
        =
        |\{S_T \in \mathcal{I}_{m-l}(W):S_T\cap L_{T}=\{0\}\}|,
    \end{equation}
    where $W=T^{\perp}/T$ and $\dim(W)=2(n-l)$. Applying Lemma~\ref{lem:zero_intersec} with
    $m\leftarrow m-l$ and $n\leftarrow n-l$ completes the proof.
\end{proof}

Using Lemma~\ref{lem:iso_lag_intersection}, we can compute higher moments of the intersection size with a fixed Lagrangian $L$ for a uniformly random $S \in \mathcal{I}_m(V)$, as stated below.

\begin{lemma}\label{lem:iso_Lag_high_inter}
    For a fixed Lagrangian $L\subset V\coloneq\mathbb{F}_2^{2n}$ and $c\in \mathbb{N}$, 
    \begin{equation}
        \mathbb{E}_{S\in \mathcal{I}_{n-t}(V)}[|S\cap L|^c]
        =
        \prod_{i=1}^{c}\frac{2^{-c+i}+2^{-t}}{2^{-c+i}+2^{-n}}.
    \end{equation}
\end{lemma}
\begin{proof}    
    \begin{align}
        \mathbb{E}_{S\in \mathcal{I}_{n-t}(V)}[|S\cap L|^c] 
        &=\sum_l^{n-t} \Pr(\mathrm{dim}(S \cap L)=l)\,2^{cl}\\
        &=\frac{1}{\binom{n}{t}_2\prod_{i=t+1}^{n}(2^i+1)}\sum_{l=0}^{n-t} \binom{n}{l}_2 \binom{n-l}{n-t-l}_2 2^{(n-t-l)(n-t-l+1)/2} 2^{(n-t-l)t}2^{cl}\\
        &=\frac{1}{\prod_{i=t+1}^{n}(2^i+1)}\sum_{l=0}^{n-t} \binom{n-t}{l}_2 2^{(n-t-l)(n-t-l+1)/2}2^{(n-t-l)t}2^{cl}\\
        &=\frac{1}{\prod_{i=t+1}^{n}(2^i+1)}\sum_{l=0}^{n-t}\binom{n-t}{l}_2 2^{l(l+1)/2}2^{tl}2^{c(n-t-l)}\\
        &=2^{c(n-t)}\frac{\prod_{i=0}^{n-t-1} (2^{i+t+1-c}+1)}{\prod_{i=t+1}^{n}(2^i+1)}\\
        &=2^{c(n-t)}\frac{\prod_{i=t-c+1}^{t} (2^i+1)}{\prod_{i=n-c+1}^{n}(2^i+1)}\\
        &=\prod_{i=1}^{c}\frac{2^{-c+i}+2^{-t}}{2^{-c+i}+2^{-n}}
    \end{align}
    as desired. The second line follows from Lemma~\ref{lem:number-of-isotropic-subspaces} and Lemma~\ref{lem:iso_lag_intersection}, the third from
    Eq.~\eqref{prop:gaissian_binom_one}, the fourth by the change of variables
    $l \leftarrow n-t-l$, and the fifth from Eq.~\eqref{prop:gaissian_binom}.
\end{proof}

These higher-moment estimates will be used to control tail bounds for averages over $S \in \mathcal{I}_{n-t}(V)$. In particular, for $c=O(1)$, the corresponding higher moments remain $O(1)$.

\begin{corollary}\label{cor:f_t}
    \begin{equation}
    \mathbb{E}_{S\in \mathcal{I}_{n-t}(V)}|S\cap L|^2
    =
    1+\frac{3}{2^t}+\frac{2}{2^{2t}}-O(2^{-n}).
    \end{equation}
\end{corollary}
\begin{proof}
    Substituting $c=2$ into Lemma~\ref{lem:iso_Lag_high_inter}, we obtain
    \begin{align}
        \mathbb{E}_{S\in \mathcal{I}_{n-t}(V)}|S\cap L|^2
        &=
        \frac{(2^{-1}+2^{-t})(1+2^{-t})}{(2^{-1}+2^{-n})(1+2^{-n})} \\
        &=
        \frac{(2^{-1}+2^{-t})(1+2^{-t})}{2^{-1}}\bigl(1-O(2^{-n})\bigr) \\
        &=
        1+\frac{3}{2^t}+\frac{2}{2^{2t}}-O(2^{-n}),
    \end{align}
    as claimed.
\end{proof}

\begin{lemma}\label{lem:dim_intersection}
    Let $L_1, L_2, L_1', L_2' \subset V\coloneq\mathbb{F}_2^{2n}$ be Lagrangian subspaces such that
    \begin{equation}
        \dim(L_1 \cap L_2)=\dim(L_1' \cap L_2').
    \end{equation}
    Then, for uniformly sampled isotropic subspace $S \in \mathcal{I}_m(V)$,
    \begin{equation}
        \mathbb{E}_S \, \bigl[|S\cap L_1|^2 |S\cap L_2|^2\bigr]
        =
        \mathbb{E}_S \, \bigl[|S\cap L_1'|^2 |S\cap L_2'|^2\bigr].
    \end{equation}
\end{lemma}

\begin{proof}
    Let
    \begin{equation}
        \ell := \dim(L_1 \cap L_2)=\dim(L_1' \cap L_2').
    \end{equation}
    Choose symplectic bases
    \begin{equation}
        \{e_1,\dots,e_n,f_1,\dots,f_n\}
        \qquad\text{and}\qquad
        \{e_1',\dots,e_n',f_1',\dots,f_n'\}
    \end{equation}
    of $V$ such that
    \begin{align}
        L_1&=\langle e_1,\dots,e_n\rangle, \\
        L_2&=\langle e_1,\dots,e_\ell,f_{\ell+1},\dots,f_n\rangle, \\
        L_1'&=\langle e_1',\dots,e_n'\rangle, \\
        L_2'&=\langle e_1',\dots,e_\ell',f_{\ell+1}',\dots,f_n'\rangle.
    \end{align}
    These bases satisfy
    \begin{equation}
        [e_i,e_j]=[f_i,f_j]=0, [e_i,f_j]=\delta_{ij},
    \end{equation}
    and similarly,
    \begin{equation}
        [e_i',e_j']=[f_i',f_j']=0, [e_i',f_j']=\delta_{ij},
    \end{equation}
    for all $i,j\in[n]$.
    
    Define a linear map $\phi:V\to V$ by
    \begin{equation}
        \phi(e_i)=e_i', \phi(f_i)=f_i',
    \end{equation}
    for each $i\in[n]$. Then for any $x,y\in V$,
    \begin{equation}
        [x,y]=[\phi(x),\phi(y)],
    \end{equation}
    so $\phi \in \mathrm{Sp}(V)$. By construction,
    \begin{equation}
        \phi(L_1)=L_1', \phi(L_2)=L_2'.
    \end{equation}
    
    Since $\phi$ preserves the symplectic form, it maps isotropic subspaces to isotropic
    subspaces and hence induces a bijection on $\mathcal{I}_m(V)$. Therefore,
    \begin{align}
        \mathbb{E}_S \, [|S\cap L_1'|^2 |S\cap L_2'|^2]
        &=
        \mathbb{E}_S \, [|S\cap \phi(L_1)|^2 |S\cap \phi(L_2)|^2] \\
        &=
        \mathbb{E}_S \, [|\phi^{-1}(S)\cap L_1|^2 |\phi^{-1}(S)\cap L_2|^2] \\
        &=
        \mathbb{E}_S \, [|S\cap L_1|^2 |S\cap L_2|^2],
    \end{align}
    as desired.
\end{proof}

We next state two lemmas that will serve as key ingredients in the proof.

\begin{lemma}\label{lem:mu}
    Let $L_1, L_2\subset V\coloneq\mathbb{F}_2^{2n}$ be Lagrangian subspaces such that
    $K\coloneq L_1\cap L_2$ with $\dim(K)=\kappa$. Then, for uniformly sampled
    $S\sim \mathcal{I}_m(V)$,
    \begin{align}
        \mu(n,\kappa,m)
        &\coloneq
        \mathbb{E}_{S}\bigl[\, |S\cap L_2|^2 \bigm| S\cap L_1=\{0\} \,\bigr] \\
        &=
        1+3(1-2^{\kappa-n})\frac{2^m-1}{2^n-1} \notag\\
        &\quad
        +2(1-2^{\kappa-n})(1-2^{\kappa-n+1})
        \frac{(2^m-1)(2^{m-1}-1)}{(2^n-1)(2^{n-1}-1)} \\
        &\le
        1+3\cdot 2^{m-n}+2^{2m-2n+1}.
    \end{align}
\end{lemma}

\begin{proof}
    Since
    \begin{equation}
        |S\cap L_2|=\sum_{x\in L_2}\mathbf{1}\{x\in S\},
    \end{equation}
    we have
    \begin{align}
        |S\cap L_2|^2
        &=
        \sum_{x,y\in L_2}\mathbf{1}\{x\in S\}\mathbf{1}\{y\in S\} \\
        &=
        \sum_{x,y\in L_2}\mathbf{1}\{\langle x,y\rangle\subseteq S\}.
    \end{align}
    Therefore,
    \begin{align}
        \mu(n,k,m)
        &=
        \sum_{x,y\in L_2}
        \Pr_S\!\bigl(\langle x,y\rangle\subseteq S \,\bigm|\, S\cap L_1=\{0\}\bigr).
    \end{align}
    We group the pairs $(x,y)\in L_2\times L_2$ according to the dimension
    $d\in\{0,1,2\}$ of the subspace $U_d\coloneq \langle x,y\rangle$.
    We also use the fact that if any one of $x$, $y$, or $x+y$ lies in $L_1\cap L_2=K$, then the corresponding term is zero, and that the conditional probability above depends only on $d$, not on the choice of $(x,y)$.
    For $d\in\{1,2\}$, let $U_d\subseteq L_2$ be any fixed $d$-dimensional subspace satisfying
    $U_d\cap L_1=\{0\}$, and define
    \begin{equation}
        p_d
        :=
        \Pr_S\!\bigl(U_d\subseteq S \,\bigm|\, S\cap L_1=\{0\}\bigr).
    \end{equation}
    Then
    \begin{equation}\label{eq:mu_expand_refined}
        \mu(n,\kappa,m)
        =
        1+3(2^n-2^\kappa)p_1+(2^n-2^\kappa)(2^n-2^{\kappa+1})p_2.
    \end{equation}

    It remains to compute $p_d$. By definition,
    \begin{equation}
        p_d
        =
        \frac{
            \bigl|\{S\in \mathcal{I}_m(V): U_d\subseteq S,\ S\cap L_1=\{0\}\}\bigr|
        }{
            \bigl|\{S\in \mathcal{I}_m(V): S\cap L_1=\{0\}\}\bigr|
        }.
    \end{equation}
    The denominator is given by Lemma~\ref{lem:zero_intersec}:
    \begin{equation}\label{eq:mu_denominator}
        \bigl|\{S\in \mathcal{I}_m(V): S\cap L_1=\{0\}\}\bigr|
        =
        \binom{n}{m}_2
        2^{m(m+1)/2}
        2^{m(n-m)}.
    \end{equation}

    For the numerator, since $U_d$ is isotropic, the quotient
    \begin{equation}
        W:=U_d^\perp/U_d
    \end{equation}
    is a symplectic vector space of dimension $2(n-d)$, and the map
    $S\mapsto S/U_d$ gives a bijection between
    \[
        \{S\in \mathcal{I}_m(V): U_d\subseteq S\}
        \quad\text{and}\quad
        \mathcal{I}_{m-d}(W).
    \]
    Set
    \begin{equation}
        \overline{L}_1:=(L_1 \cap U_d^{\perp}+U_d)/U_d \subseteq W.
    \end{equation}
    Since $L_1$ is Lagrangian and $U_d\cap L_1=\{0\}$, the subspace $\overline{L}_1$
    is Lagrangian in $W$. Moreover, for every isotropic $S$ containing $U_d$,
    \begin{equation}
        S\cap L_1=\{0\}
        \qquad\Longleftrightarrow\qquad
        \overline{S}\cap \overline{L}_1=\{0\},
    \end{equation}
    where $\overline{S}\coloneq S/U_d$.

    Indeed, since $U_d\subseteq S$,
    \begin{align}
        \overline{S} \cap \overline{L}_1 
        &= (S/U_d) \cap ((L_1 \cap U_d^{\perp}+U_d)/U_d) \\
        &= (S\cap(L_1 \cap U_d^{\perp}+U_d))/U_d \\
        &= (S \cap L_1 + U_d)/U_d.
    \end{align}
    If $S\cap L_1=\{0\}$, then clearly $\overline{S}\cap \overline{L}_1=\{0\}$.
    Conversely, if $\overline{S}\cap \overline{L}_1=\{0\}$, then
    $S\cap L_1 \subseteq U_d$, and since $L_1\cap U_d=\{0\}$, it follows that
    $S\cap L_1=\{0\}$.

    Therefore,
    \begin{align}
        &\bigl|\{S\in \mathcal{I}_m(V): U_d\subseteq S,\ S\cap L_1=\{0\}\}\bigr| \\
        &= \bigl|\{\overline{S}\in \mathcal{I}_{m-d}(W): \overline{S}\cap \overline{L}_1=\{0\}\}\bigr|.
    \end{align}
    Applying Lemma~\ref{lem:zero_intersec} in $W$, with $n$ replaced by $n-d$
    and $m$ replaced by $m-d$, we obtain
    \begin{equation}\label{eq:mu_numerator}
        \bigl|\{S\in \mathcal{I}_m(V): U_d\subseteq S,\ S\cap L_1=\{0\}\}\bigr|
        =
        \binom{n-d}{m-d}_2
        2^{(m-d)(m-d+1)/2}
        2^{(m-d)(n-m)}.
    \end{equation}
    Combining \eqref{eq:mu_denominator} and \eqref{eq:mu_numerator}, we obtain
    \begin{equation}\label{eq:pd_general_refined}
        p_d
        =
        2^{-dn+d(d-1)/2}
        \frac{\binom{n-d}{m-d}_2}{\binom{n}{m}_2}.
    \end{equation}
    In particular,
    \begin{align}
        p_1&=2^{-n}\frac{2^m-1}{2^n-1},\\
        p_2&=2^{-2n+1}\frac{(2^m-1)(2^{m-1}-1)}{(2^n-1)(2^{n-1}-1)}.
    \end{align}

    Substituting these into \eqref{eq:mu_expand_refined} gives
    \begin{align}
        \mu(n,\kappa,m)
        &=
        1+3(1-2^{\kappa-n})\frac{2^m-1}{2^n-1} \\
        &\quad
        +2(1-2^{\kappa-n})(1-2^{\kappa-n+1})
        \frac{(2^m-1)(2^{m-1}-1)}{(2^n-1)(2^{n-1}-1)} \\
        &\le
        1+3\cdot 2^{m-n}+2^{2m-2n+1}.
    \end{align}
    This proves the claim.
\end{proof}
\begin{lemma}\label{lem:gamma}
    Let $L_1, L_2\subset V\coloneq\mathbb{F}_2^{2n}$ be Lagrangian subspaces such that
    $K\coloneq L_1\cap L_2$ with $\dim(K)=\kappa$. Then, 
    \begin{align}
        \gamma(n,m,\kappa)
        &\coloneq
        \mathbb{E}_{S\sim \mathcal{I}_{m}(L_1)}\left[|S\cap K|^2 \right]\\
        &=
        1+3(2^{\kappa}-1)\frac{2^m-1}{2^n-1}
        +(2^{\kappa}-1)(2^{\kappa}-2)\frac{(2^m-1)(2^{m-1}-1)}{(2^n-1)(2^{n-1}-1)}\\
        &\le 1+3\cdot2^{\kappa+m-n}+4^{\kappa+m-n}.
    \end{align}
\end{lemma}

\begin{proof}
    Since
    \begin{equation}
        |S \cap K|=\sum_{x\in K}\mathbf{1}\{x\in S\},
    \end{equation}
    we have
    \begin{align}
        |S \cap K|^2
        &=\sum_{x,y \in K}\mathbf{1}\{x\in S\}\mathbf{1}\{y\in S\}\\
        &=1+3\sum_{x \in K\setminus \{0\}}\mathbf{1}\{x\in S\}
        +\sum_{x \neq y \in K\setminus \{0\}}\mathbf{1}\{x\in S\}\mathbf{1}\{y\in S\}.
    \end{align}
    Taking expectation, we obtain
    \begin{align}
        \gamma(n,m,k)
        &=
        1+3\sum_{x \in K\setminus \{0\}}
        \mathbb{E}_{S\sim \mathcal{I}_m(L_1)} \left[\mathbf{1}\{x\in S\}\right] \\
        &\quad
        +\sum_{x \neq y \in K\setminus \{0\}}
        \mathbb{E}_{S\sim \mathcal{I}_m(L_1)}
        \left[\mathbf{1}\{x\in S\}\mathbf{1}\{y\in S\}\right].
    \end{align}
    Since the probability depends only on the dimension of the span,
    this becomes
    \begin{align}
        \gamma(n,m,\kappa)
        &=
        1+3(2^\kappa-1)\Pr_S (\langle x \rangle \subseteq S) \\
        &\quad +(2^\kappa-1)(2^\kappa-2)\Pr_S(\langle x , y\rangle \subseteq S),
        \label{eq:main2}
    \end{align}
    where $x\neq y \in K \setminus \{0\}$.

    Now let $U_d\subseteq L_1$ be a fixed $d$-dimensional subspace. The number of
    $m$-dimensional subspaces of $L_1$ containing $U_d$ is
    $\binom{n-d}{m-d}_2$, and the total number of $m$-dimensional subspaces of $L_1$
    is $\binom{n}{m}_2$. Therefore,
    \begin{equation}
        \Pr_S(U_d \subseteq S)
        =
        \frac{\binom{n-d}{m-d}_2}{\binom{n}{m}_2}.
    \end{equation}
    In particular,
    \begin{align}
        \Pr_S(U_1 \subseteq S)
        &=\frac{2^m-1}{2^n-1},\\
        \Pr_S(U_2 \subseteq S)
        &=\frac{(2^m-1)(2^{m-1}-1)}{(2^n-1)(2^{n-1}-1)}.
    \end{align}
    Substituting these into Eq.~\eqref{eq:main2} gives
    \begin{align}
        \gamma(n,m,\kappa)
        &=
        1+3(2^{\kappa}-1)\frac{2^m-1}{2^n-1}
        +(2^{\kappa}-1)(2^{\kappa}-2)\frac{(2^m-1)(2^{m-1}-1)}{(2^n-1)(2^{n-1}-1)}\\
        &\le 1+3\cdot2^{\kappa+m-n}+4^{\kappa+m-n}.
    \end{align}
    This proves the claim.
\end{proof}

\subsection{Part I: Recovery of the Weyl support}

In the first part, we show that, when $t=O(\log n)$, with high probability over independently sampled
$C_1,\dots,C_m \sim \mathcal{C}_k$, the resulting measurement bases are sufficient to recover the full stabilizer group. Concretely, we show that
\begin{equation}
    \sum_{i=1}^{m} \bigl(\mathrm{Weyl}(\rho)\cap C_i^{\dagger}(\mathcal{Z})\bigr)
    =
    \mathrm{Weyl}(\rho)
\end{equation}
holds with high probability.

This conclusion does not hold in the worst case for non-adaptive shallow-depth measurements. For example, as noted in the main text, $\mathrm{Weyl}(\rho_{\mathrm{GHZ}})$ may fail to be recovered from polynomially many sampled bases. However, such worst-case obstructions occur only for an exponentially small fraction of $\mathrm{Weyl}(\rho)\in\mathcal{I}_{n-t}(V)$. For almost all target subspaces with $t=O(\log n)$, the following lemma provides the required recovery guarantee.

\begin{lemma}[Intersection with $\mathcal{Z}$]\label{lem:pr_lower_overlap}
    Assume that $t=O(\log n)$ and $k=\Omega(\log n)$. Let $C$ be sampled independently from $\mathcal{C}_k$. Then, for all but a
    $2^{-\Omega(n)}$ fraction of $S \in \mathcal{I}_{n-t}(V)$, the following holds:
    for every codimension-one subspace $T \subset S$,
    \begin{equation}
        \Pr_{C}\!\left(C(S \setminus T) \cap \mathcal{Z} \neq \emptyset\right)
        \ge \Omega(2^{-t}).
    \end{equation}
\end{lemma}
\begin{proof}
    Define $X_C=X_C(S,T):=|C(S \setminus T)\cap \mathcal{Z}|$. Then
    \begin{equation}
        \Pr_C\!\left(C(S \setminus T)\cap \mathcal{Z} \neq \emptyset\right)
        =\Pr_C(X_C>0).
    \end{equation}
    Since $X_C$ is nonnegative and finite, Paley--Zygmund's inequality gives
    \begin{equation}
        \Pr_C(X_C>0)
        \ge 
        \frac{\mathbb{E}_C[X_C]^2}{\mathbb{E}_C[X_C^2]}.
    \end{equation}
    For the numerator, Lemma~\ref{lem:E_XU} gives $\mathbb{E}_C[X_C]\ge \Omega(2^{-t})$, and hence
    \begin{equation}\label{eq:E_XU}
        \mathbb{E}_C[X_C]^2 \ge \Omega(2^{-2t}).
    \end{equation}
    For the denominator, Corollary~\ref{cor:EXU2_exp_small_failure} implies that, for all but a $2^{-\Omega(n)}$ fraction of $S\in \mathcal{I}_{n-t}(V)$,
    \begin{equation}\label{eq:E_XU^2}
        \mathbb{E}_C[X_C^2]\le O(2^{-t}).
    \end{equation}
    Combining Eq.~\eqref{eq:E_XU} and Eq.~\eqref{eq:E_XU^2} proves the claim.
\end{proof}

\begin{corollary}[Recovery of the Weyl support]\label{cor:weyl-support-recovery}
    Assume that $t=O(\log n)$ and $k=\Omega(\log n)$. Let
    $C_1,\dots,C_m$ be sampled independently from $\mathcal{C}_k$, where
    $m=O(n2^t)$. Then, with high probability over the choice of
    $C_1,\dots,C_m$, the following holds for all but a $2^{-\Omega(n)}$
    fraction of $S\in\mathcal{I}_{n-t}(V)$:
    \begin{equation}
        \sum_{i=1}^{m} \bigl(S \cap C_i^{\dagger}(\mathcal{Z})\bigr)=S.
    \end{equation}
\end{corollary}

\begin{proof}
    Fix $S\in\mathcal{I}_{n-t}(V)$ for which Lemma~\ref{lem:pr_lower_overlap}
    holds for every codimension-one subspace of $S$. For each $i\in[m]$, define
    \begin{equation}
        S_i:=S\cap C_i^{\dagger}(\mathcal{Z}),
        \qquad
        T_i:=\sum_{j=1}^{i} S_j,
    \end{equation}
    with the convention $T_0=\{0\}$.

    Suppose that $T_{i-1}\subsetneq S$. Since $T_{i-1}$ is a proper subspace of $S$,
    there exists a codimension-one subspace $T\subset S$ such that
    $T_{i-1}\subseteq T$. If
    \begin{equation}
        C_i(S\setminus T)\cap\mathcal{Z}\neq\emptyset,
    \end{equation}
    then there exists $x\in S\setminus T$ such that
    $x\in S\cap C_i^{\dagger}(\mathcal{Z})=S_i$. In particular, $x\notin T_{i-1}$, and hence
    $T_i\supsetneq T_{i-1}$. Therefore,
    \begin{align}
        \Pr\!\left(T_i\supsetneq T_{i-1}\mid T_{i-1}\subsetneq S\right)
        &\ge
        \Pr\!\left(C_i(S\setminus T)\cap\mathcal{Z}\neq\emptyset\right)\\
        &\ge \Omega(2^{-t}),
    \end{align}
    where the last line follows from Lemma~\ref{lem:pr_lower_overlap}.

    Thus, as long as $T_{i-1}\neq S$, the dimension increases by at least one with
    probability at least $\Omega(2^{-t})$. Since $\dim(S)=n-t$, it follows by a
    standard Chernoff-bound argument that after $m=O(n2^t)$ independent samples,
    we have $\dim(T_m)=n-t$ with high probability. Because $T_m\subseteq S$, this implies
    $T_m=S$, i.e.,
    \begin{equation}
        \sum_{i=1}^{m} \bigl(S \cap C_i^{\dagger}(\mathcal{Z})\bigr)=S.
    \end{equation}
    This proves the claim.
\end{proof}

In the proof of Lemma~\ref{lem:pr_lower_overlap}, we use a lower bound on the first moment of $X_C$ and an upper bound on its second moment. We state these bounds in the following lemmas.
\begin{lemma}\label{lem:E_XU}
    Let $S\in \mathcal{I}_{n-t}(V)$ and let $T\subset S$ be a codimension-one subspace. Define
    $X_C:=|C(S \setminus T) \cap \mathcal{Z}|$, where $C$ is sampled from
    $\mathcal{C}_k$ with $k=\Omega(\log n)$. Then
    \begin{equation}
        \mathbb{E}_C[X_C]\ge \Omega(2^{-t}).
    \end{equation}
\end{lemma}

\begin{proof}
    Using the Pauli collision probability $m_P$ and Eq.~\eqref{eq:mp_lower}, we obtain
    \begin{align}
        \mathbb{E}_C[X_C]
        &= \mathbb{E}_C\bigl[|C(S \setminus T)\cap \mathcal{Z}|\bigr] \\
        &= \mathbb{E}_C\left[\sum_{P \in S\setminus T} \mathbf{1}\{C(P)\in \mathcal{Z}\}\right] \\
        &= \sum_{P \in S\setminus T} \mathbb{E}_C \bigl[\mathbf{1}\{C(P)\in \mathcal{Z}\}\bigr] \\
        &= \sum_{P\in S\setminus T} m_P \\
        &\ge 2^{n-t-1}\min_{P} m_P \\
        &= \Omega(2^{-t}).
    \end{align}
    The final bound follows from Eq.~\eqref{eq:mp_lower}.
\end{proof}

\begin{fact}[Chebyshev's inequality and a simple consequence]\label{fact:chebyshev}
Let $X$ be a real-valued random variable with finite variance and
$\mathbb{E}[X]>0$. Then, for any $a>0$,
\[
    \Pr\!\left(|X-\mathbb{E}[X]|\ge a\right)
    \le
    \frac{\mathrm{Var}(X)}{a^2}.
\]

In particular, for any constant $c>1$,
\[
    \Pr\!\left(X \ge c\,\mathbb{E}[X]\right)
    \le
    \Pr\!\left(|X-\mathbb{E}[X]|\ge (c-1)\mathbb{E}[X]\right)
    \le
    \frac{\mathrm{Var}(X)}{(c-1)^2\mathbb{E}[X]^2}.
\]
Therefore,
\[
    \Pr\!\left(X < c\,\mathbb{E}[X]\right)
    \ge
    1-\frac{\mathrm{Var}(X)}{(c-1)^2\mathbb{E}[X]^2}.
\]
Equivalently, for any fixed constant $c>1$,
\[
    \Pr\!\left(X < O(\mathbb{E}[X])\right)
    \ge
    1-O\!\left(\frac{\mathrm{Var}(X)}{\mathbb{E}[X]^2}\right).
\]
\end{fact}

\begin{lemma}\label{lem:E_XU^2}
    Let $S\in \mathcal{I}_{n-t}(V)$, and let $T\subset S$ be a codimension-one subspace. Define $X_C = \left|C(S\setminus T)\cap \mathcal{Z}\right|$, where $C$ is sampled from $\mathcal{C}_k$. Then
    \begin{equation}
        \Pr_S\!\left(\mathbb{E}_C[X_C^2] \le O(2^{-t})\right)
        \ge 1-O\left(2^{2t}\left(\frac{2}{2^k+1}\right)^{n/k}\right).
    \end{equation}
\end{lemma}

\begin{proof}
    We first bound $\mathbb{E}_C[X_C^2]$ as
    \begin{align}
        \mathbb{E}_C[X_C^2]
        &=
        \mathbb{E}_C \bigl[|C(S\setminus T) \cap \mathcal{Z}|^2\bigr] \\
        &\le
        \mathbb{E}_C \bigl[|C(S) \cap \mathcal{Z}|^2\bigr]-1.
    \end{align}
    Let $\beta(S)\coloneq \mathbb{E}_C [|C(S)\cap \mathcal{Z}|^2]$. Then the following hold:
    \begin{enumerate}
        \item $\mathbb{E}_S[\beta(S)]
        =
        f_t-O(2^{-n})$ by Lemma~\ref{lem:beta},
        \item $\mathbb{E}_S[\beta(S)^2]
        \le
        f_t^2+O\left(\left(\frac{2}{2^k+1}\right)^{n/k}\right)$ by Lemma~\ref{lem:beta^2},
    \end{enumerate}
    where $f_t=1+3/2^t+2/2^{2t}$. It follows that
    \begin{equation}
        \mathrm{Var}[\beta(S)]
        \le
        f_t^2+O\left(\left(\frac{2}{2^k+1}\right)^{n/k}\right)-\left(f_t-O(2^{-n})\right)^2
        =
        O\left(\left(\frac{2}{2^k+1}\right)^{n/k}\right).
    \end{equation}
    Since $f_t-1=\Theta(2^{-t})$, Fact~\ref{fact:chebyshev} gives
    \begin{align}
        \Pr_S\!\left(\mathbb{E}_C[X_C^2]\le O(2^{-t})\right)
        &\ge \Pr_S\!\left(\beta(S)-1<O(2^{-t})\right)\\
        &\ge 1-O\bigl(2^{2t}\mathrm{Var}[\beta(S)]\bigr)\\
        &\ge 1-O\left(2^{2t}\left(\frac{2}{2^k+1}\right)^{n/k}\right).
    \end{align}
    This proves the claim.
\end{proof}

\begin{corollary}\label{cor:EXU2_exp_small_failure}
    Assume that $t=O(\log n)$.
    Let $S\in \mathcal{I}_{n-t}(V)$, and let $T\subset S$ be a codimension-one subspace.
    Define $X_C = \left|C(S\setminus T) \cap \mathcal{Z}\right|$, where $C$ is sampled from $\mathcal{C}_k$. Then
    \begin{equation}\label{eq:main_inequal}
        \Pr_S\!\left(\mathbb{E}_C[X_C^2] \le O(2^{-t})\right)
        >
        1-2^{-\Omega(n)}.
    \end{equation}
    Note that Eq.~\eqref{eq:main_inequal} holds for all $k\in\mathbb[n]$.
\end{corollary}

\begin{lemma}\label{lem:beta}
    Let $S \in \mathcal{I}_{n-t}(V)$ and let $\beta(S)=\mathbb{E}_C[|C(S)\cap \mathcal{Z}|^2]$, where $C$ is sampled from $\mathcal{C}_k$. Then
    \begin{equation}
        \mathbb{E}_S[\beta(S)]=f_t-O(2^{-n}),
    \end{equation}
    where $f_t=1+3/2^{t}+2/2^{2t}$.
\end{lemma}
\begin{proof}
    For any fixed Clifford circuit $C$, the map $S \mapsto CSC^{\dagger}$ is a bijection on $\mathcal{I}_{n-t}(V)$. Hence, if $S$ is sampled uniformly from $\mathcal{I}_{n-t}(V)$, then $CSC^{\dagger}$ has the same distribution as $S$. Since every $C \in \mathcal{C}_k$ is Clifford, we obtain
    \begin{align}
        \mathbb{E}_S[\beta(S)]
        &= \mathbb{E}_S \mathbb{E}_C[|C(S)\cap \mathcal{Z}|^2] \\
        &= \mathbb{E}_C \mathbb{E}_S[|C(S)\cap \mathcal{Z}|^2] \\
        &= \mathbb{E}_C \mathbb{E}_S[|S\cap \mathcal{Z}|^2] \\
        &= \mathbb{E}_S[|S\cap \mathcal{Z}|^2]\\
        &=f_t-O(2^{-n}).
    \end{align}
    In the last line, since $\mathcal{Z}$ is Lagrangian, the claim follows from Corollary~\ref{cor:f_t}.
\end{proof}

\begin{lemma}\label{lem:main_chunk}
    Let $A,B \subset V\coloneq\mathbb{F}_2^{2n}$ be Lagrangian subspaces, and let
    $K\coloneq A\cap B$ with $\dim(K)=\kappa$.
    If $S$ is sampled uniformly from $\mathcal{I}_{n-t}(V)$, then
    \begin{equation}
        \mathbb{E}_S\bigl[|S\cap A|^2 |S\cap B|^2\bigr]
        \le
        f_t^2 + O(2^{\kappa-n}) + O(4^{\kappa-n}).
    \end{equation}
\end{lemma}

\begin{proof}
    We first condition on the dimension of $S\cap A$. Writing
    \begin{equation}
        a:=\dim(S\cap A),
    \end{equation}
    we obtain
    \begin{align}
        \mathbb{E}_S\bigl[|S\cap A|^2 |S\cap B|^2\bigr]
        &=
        \sum_{a=0}^{n-t}
        \Pr_S\!\bigl(\dim(S\cap A)=a\bigr)\,
        \mathbb{E}_S\!\bigl[
            |S\cap A|^2 |S\cap B|^2
            \,\bigm|\,
            \dim(S\cap A)=a
        \bigr] \\
        &=
        \sum_{a=0}^{n-t}
        \Pr_S\!\bigl(\dim(S\cap A)=a\bigr)\,
        2^{2a}\,
        \mathbb{E}_S\!\bigl[
            |S\cap B|^2
            \,\bigm|\,
            \dim(S\cap A)=a
        \bigr].
        \label{eq:E_a_refined}
    \end{align}

    We next analyze the conditional expectation inside the sum.
    Conditioning further on the actual subspace $M=S\cap A$, we have
    \begin{equation}
        \mathbb{E}_S\!\bigl[
            |S\cap B|^2
            \,\bigm|\,
            \dim(S\cap A)=a
        \bigr]
        =
        \mathbb{E}_{M\sim \mathcal{I}_a(A)}
        \left[
            \mathbb{E}_S\!\bigl[
                |S\cap B|^2
                \,\bigm|\,
                S\cap A=M
            \bigr]
        \right].
        \label{eq:main_chunk_condition_on_M}
    \end{equation}
    Here we used that, conditional on $\dim(S\cap A)=a$, the intersection
    $M=S\cap A$ is uniformly distributed over the $a$-dimensional subspaces of $A$.

    Fix such a subspace $M\subset A$ with $\dim(M)=a$, and define
    \begin{equation}
        W\coloneq M^\perp/M,
        \qquad
        S_M\coloneq S/M,
        \qquad
        A_M\coloneq A/M,
        \qquad
        B_M\coloneq (B\cap M^\perp + M)/M.
    \end{equation}
    \begin{figure}[t]
        \centering
        \includegraphics[width=0.6\linewidth, trim=1cm 9cm 3cm 0cm, clip]{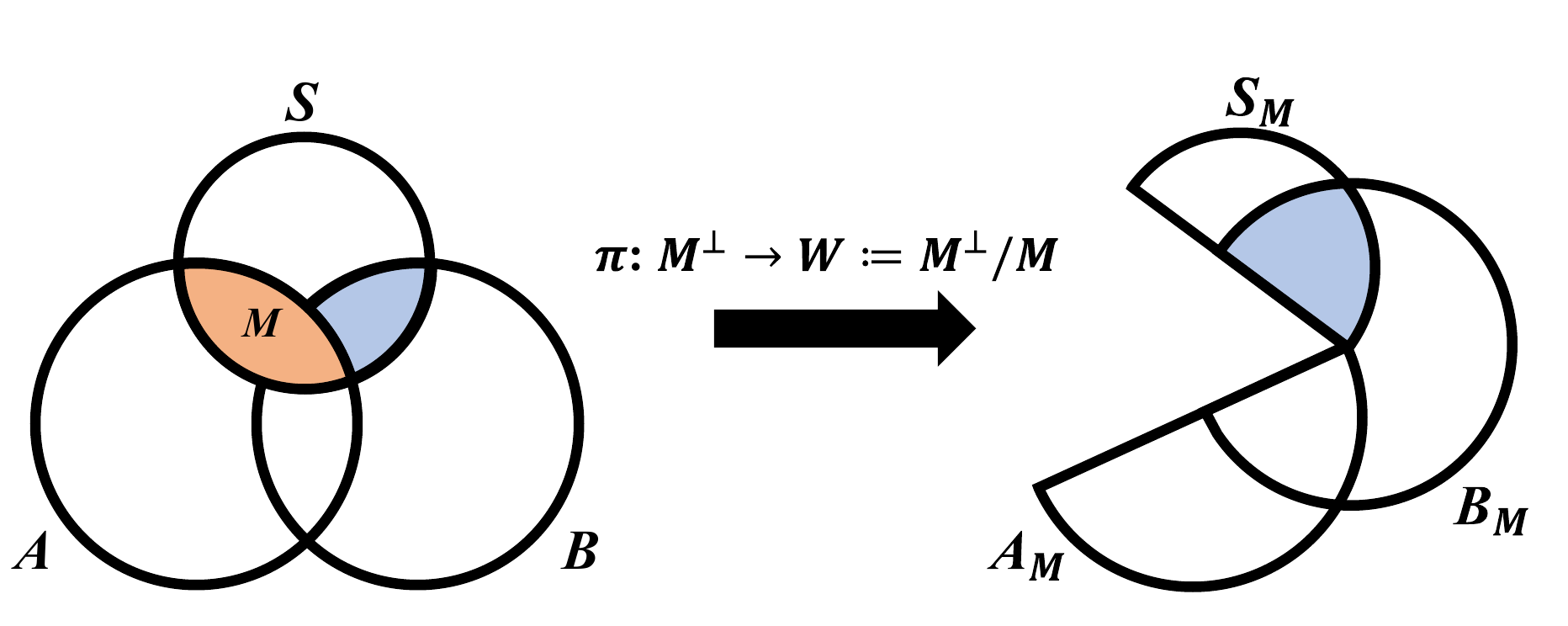}
        \caption{\label{fig:3}
        Quotient reduction.
        Left: for a fixed isotropic subspace $S$, let $M=S\cap A$.
        Right: after passing to the quotient space $M^\perp/M$, the subspaces $S_M=S/M$, $A_M=A/M$, and $B_M=(B\cap M^\perp+M)/M$ encode the remaining intersection structure. In particular, conditioning on $S\cap A=M$ reduces the analysis of $S\cap B$ to that of $S_M\cap B_M$ in the quotient symplectic space.}
    \end{figure}
    Since $A$ and $B$ are Lagrangian subspaces of $V$, both $A_M$ and $B_M$ are Lagrangian subspaces of the symplectic space $W$~[Fig.~\ref{fig:3}]. Applying the first isomorphism theorem to the restriction of the quotient map $\pi:M^\perp\to W$ to $S\cap B$, we obtain
    \begin{equation}
        (S\cap B)/(M \cap K)
        \cong
        (S\cap B+M)/M=S_M \cap B_M.
    \end{equation}
    Consequently,
    \begin{equation}
        |S\cap B|
        =
        |M\cap K|\,|S_M\cap B_M|.
    \end{equation}
    Therefore,
    \begin{align}
        \mathbb{E}_S\!\bigl[
            |S\cap B|^2
            \,\bigm|\,
            S\cap A=M
        \bigr]
        &=
        |M\cap K|^2\,
        \mathbb{E}_S\!\bigl[
            |S_M\cap B_M|^2
            \,\bigm|\,
            S\cap A=M
        \bigr].
    \end{align}
    Indeed, the map $S\mapsto S/M$ gives a bijection between $\{S\in\mathcal{I}_{n-t}(V):S \cap A =M\}$ and $\{S_M\in\mathcal{I}_{n-t-a}(W):S_M \cap A_M =\{0\}\}$. Since the original distribution is uniform, the induced distribution on the quotient is also uniform.
    Hence
    \begin{equation}
        \mathbb{E}_S\!\bigl[
            |S_M\cap B_M|^2
            \,\bigm|\,
            S\cap A=M
        \bigr]
        =
        \mathbb{E}_{S_M\sim \mathcal{I}_{n-t-a}(W)}
        \!\bigl[
            |S_M\cap B_M|^2
            \,\bigm|\,
            S_M\cap A_M=\{0\}
        \bigr].
    \end{equation}
    Applying Lemma~\ref{lem:mu} in the quotient space $W$ gives
    \begin{equation}
        \mathbb{E}_{S_M\sim \mathcal{I}_{n-t-a}(W)}
        \!\bigl[
            |S_M\cap B_M|^2
            \,\bigm|\,
            S_M\cap A_M=\{0\}
        \bigr]
        \le f_t,
    \end{equation}
    where $f_t=1+3/2^t+2/2^{2t}$.
    Consequently,
    \begin{equation}
        \mathbb{E}_S\!\bigl[
            |S\cap B|^2
            \,\bigm|\,
            S\cap A=M
        \bigr]
        \le
        f_t\,|M\cap K|^2.
    \end{equation}
    Substituting this bound into Eq.~\eqref{eq:main_chunk_condition_on_M}, we obtain
    \begin{align}
        \mathbb{E}_S\!\bigl[
            |S\cap B|^2
            \,\bigm|\,
            \dim(S\cap A)=a
        \bigr]
        &\le
        f_t\,
        \mathbb{E}_{M\sim \mathcal{I}_a(A)}\bigl[|M\cap K|^2\bigr].
    \end{align}
    By Lemma~\ref{lem:gamma},
    \begin{equation}
        \mathbb{E}_{M\sim \mathcal{I}_a(A)}\bigl[|M\cap K|^2\bigr]
        \le
        1+3\cdot 2^{\kappa+a-n}+4^{\kappa+a-n}.
    \end{equation}
    Therefore,
    \begin{equation}
        \mathbb{E}_S\!\bigl[
            |S\cap B|^2
            \,\bigm|\,
            \dim(S\cap A)=a
        \bigr]
        \le
        f_t\bigl(1+3\cdot 2^{\kappa+a-n}+4^{\kappa+a-n}\bigr).
    \end{equation}
    Plugging this estimate into Eq.~\eqref{eq:E_a_refined} yields
    \begin{align}
        \mathbb{E}_S\bigl[|S\cap A|^2 |S\cap B|^2\bigr]
        &\le
        f_t
        \sum_{a=0}^{n-t}
        \Pr_S\!\bigl(\dim(S\cap A)=a\bigr)
        \Bigl(
            2^{2a}
            + 3\cdot 2^{\kappa-n}2^{3a}
            + 4^{\kappa-n}2^{4a}
        \Bigr) \\
        &=
        f_t\Bigl(
            \mathbb{E}[2^{2a}]
            + O(2^{\kappa-n})\,\mathbb{E}[2^{3a}]
            + O(4^{\kappa-n})\,\mathbb{E}[2^{4a}]
        \Bigr),
    \end{align}
    where the expectation is taken with respect to the distribution of
    $a=\dim(S\cap A)$. Finally, since $A$ is Lagrangian, Lemma~\ref{lem:iso_Lag_high_inter} and
    Corollary~\ref{cor:f_t} imply that
    \begin{equation}
        \mathbb{E}[2^{2a}]
        =
        \mathbb{E}_S[|S\cap A|^2]
        =
        f_t+O(2^{-n}),
    \end{equation}
    while $\mathbb{E}[2^{3a}]$ and $\mathbb{E}[2^{4a}]$ are both $O(1)$.
    Hence
    \begin{equation}
        \mathbb{E}_S\bigl[|S\cap A|^2 |S\cap B|^2\bigr]
        \le
        f_t^2 + O(2^{\kappa-n}) + O(4^{\kappa-n}),
    \end{equation}
    as claimed.
\end{proof}

\begin{lemma}\label{lem:beta^2}
Let $S \in \mathcal{I}_{n-t}(V)$, and define $\beta(S):=\mathbb{E}_C[|C(S)\cap \mathcal{Z}|^2]$, where $C$ is sampled from $\mathcal{C}_k$. Then
\begin{equation}
    \mathbb{E}_S[\beta(S)^2]
    \le
    f_t^2+O\left(\left(\frac{2}{2^k+1}\right)^{n/k}\right),
\end{equation}
where $f_t=1+3/2^{t}+2/2^{2t}$.
\end{lemma}

\begin{proof}
    We expand the second moment as
    \begin{align}
        \mathbb{E}_{S}[\beta(S)^2]
        &=\mathbb{E}_S\mathbb{E}_{C_1}\mathbb{E}_{C_2}\bigl[|C_1(S)\cap \mathcal{Z}|^2\,|C_2(S)\cap\mathcal{Z}|^2\bigr]\\
        &=\mathbb{E}_{C_1}\mathbb{E}_{C_2}\mathbb{E}_S\bigl[|S\cap C_1^{\dagger}(\mathcal{Z})|^2\,|S\cap C_2^{\dagger}(\mathcal{Z})|^2\bigr]\\
        &=\sum_{l=0}^{n}\Pr_{C_1,C_2}\bigl(\dim(C_1^{\dagger}(\mathcal{Z}) \cap C_2^{\dagger}(\mathcal{Z}))=l\bigr)\,
        \mathbb{E}_S[|S\cap A|^2|S\cap B|^2]\\
        &=\sum_{l=0}^{n}\Pr_{C}\bigl(\dim(C^{\dagger}(\mathcal{Z}) \cap \mathcal{Z})=l\bigr)\,
        \mathbb{E}_S[|S\cap A|^2|S\cap B|^2],
    \end{align}
    where in the third line we used Lemma~\ref{lem:dim_intersection}, and $A,B\subset V$ are Lagrangian subspaces satisfying $\dim(A\cap B)=l$.
    
    Let $G_k(s)$ be the generating function defined in Definition~\ref{def:generating_func} for the block size $k$, and let $G(s):=G_k(s)^{n/k}$. Then
    \begin{equation}
        \Pr_{C}\bigl(\dim(C^{\dagger}(\mathcal{Z}) \cap \mathcal{Z})=l\bigr)=[s^l]\,G(s),
    \end{equation}
    where $[s^l]\,G(s)$ denotes the coefficient of $s^l$ in $G(s)$. In particular,
    \begin{align}
        &G(2)=G_k(2)^{n/k}=\sum_{l=0}^{n}\Pr_{C}\bigl(\dim(C^{\dagger}(\mathcal{Z}) \cap \mathcal{Z})=l\bigr)\,2^l,\\
        &G(4)=G_k(4)^{n/k}=\sum_{l=0}^{n}\Pr_{C}\bigl(\dim(C^{\dagger}(\mathcal{Z}) \cap \mathcal{Z})=l\bigr)\,4^l.
    \end{align}
    Applying Lemma~\ref{lem:main_chunk}, we obtain
    \begin{align}
        \mathbb{E}_{S}[\beta(S)^2]
        &\le f_t^2+O\left(2^{-n}G_k(2)^{n/k}\right)+O\left(4^{-n}G_k(4)^{n/k}\right)\\
        &\le f_t^2+O\left(\left(\frac{2}{2^k+1}\right)^{n/k}\right)+O\left(\left(\frac{6}{(2^k+1)(2^k+2)}\right)^{n/k}\right)\\
        &\le f_t^2+O\left(\left(\frac{2}{2^k+1}\right)^{n/k}\right).
    \end{align}
\end{proof}

\subsection{Part II: Reconstruction from computational difference samples}

In Part I, we showed that, after sampling $m=O(n2^t)$ Clifford circuits $C_1,\dots,C_m \sim \mathcal{C}_k$, one obtains measurement bases that are sufficient
to recover $\mathrm{Weyl}(\rho)$. This, however, does not by itself identify the generators
of $\mathrm{Weyl}(\rho)$. It remains to show that, in each such basis, sufficiently many
computational difference samples allow us to learn the corresponding hidden Pauli
symmetries, and that combining the resulting subspaces yields a sufficiently accurate
approximation to $\mathrm{Weyl}(\rho)$.

The role of Part II is precisely to establish this second step. For each
$\rho_i := C_i \rho C_i^\dagger$, we use computational difference sampling to construct a
subspace $\widehat{H}_i \subseteq \mathcal{X}$ that captures most of the mass of the
distribution $r_{\rho_i}$. We then show that the subspaces
$C_i^\dagger\bigl((\widehat{H}_i+\mathcal{Z})^\perp\bigr)$ can be aggregated so as to
produce a subspace $\widehat{S}$ satisfying
$\mathbb{E}_{x\sim\widehat{S}}[\mathrm{tr}(\rho W_x)^2]\ge 1-\epsilon$. In particular, this
part extends the computational difference sampling analysis from the pure-state setting~\cite{grewal2025efficient} to
general mixed-state inputs. 

\begin{fact}[Finding a heavy-weight subspace~\cite{grewal2025efficient, chen2025stabilizer}]\label{fact:heavy}
    Let $\epsilon > 0$, $m \in \mathbb{N}$, and let $\mathcal{D}$ be a distribution over $\mathbb{F}_2^d$.
    Suppose $x_1,\dots,x_m$ are $m$ i.i.d.\ samples drawn from $\mathcal{D}$.
    Fix $0 < \delta < 1$. If $m \ge \frac{2\log(1/\delta) + 2d}{\epsilon}$, then with probability at least $1-\delta$,
    \begin{equation}
        \mathcal{D}\!\left(\operatorname{span}(x_1,\dots,x_m)\right) \ge 1-\epsilon.
    \end{equation}
\end{fact}

Applying this fact to the computational difference sampling distribution $r_{\rho}$ yields the following.

\begin{corollary}\label{cor:r_mass}
    Let $x_1,\dots,x_m$ be $m$ i.i.d.\ computational difference samples drawn from $r_{\rho}$, and let $A=\langle x_1,\dots,x_m\rangle$. If
    \begin{equation}
        m\geq\frac{2\log(1/\delta) + 2n}{\epsilon},
    \end{equation}
    then with probability at least $1-\delta$,
    \begin{equation}
        r_{\rho}(A) \geq 1-\epsilon.
    \end{equation}
\end{corollary}
\begin{proof}
    This follows directly from Fact~\ref{fact:heavy}.
\end{proof}

\begin{lemma}[Sum of heavy-weight subspaces]\label{lem:sos}
    Let $S_1,S_2 \subset V\coloneq\mathbb{F}_2^{2n}$ be subspaces satisfying
    \begin{equation}
        \frac{1}{|S_i|}\sum_{x\in S_i}\mathrm{tr}(\rho W_x)^2 \ge 1-\epsilon_i,
        \qquad i\in\{1,2\}.
    \end{equation}
    Let $S:=S_1+S_2$. Then
    \begin{equation}
        \frac{1}{|S|}\sum_{x\in S}\mathrm{tr}(\rho W_x)^2
        \ge 1-(\epsilon_1+\epsilon_2).
    \end{equation}
\end{lemma}

\begin{proof}
    For a subspace $A\subset V$, define
    \begin{equation}
        P_A
        \coloneq
        \frac{1}{|A|}\sum_{x\in A} W_x\otimes \overline{W_x},
    \end{equation}
    where the bar denotes entrywise complex conjugation in the computational basis.

    We first show that $P_A$ is an orthogonal projector. Since the Weyl operators are
    defined only up to their projective phase, for any $x,y\in V$ there exists a phase
    $\gamma(x,y)$ with $|\gamma(x,y)|=1$ such that
    \begin{equation}
        W_xW_y=\gamma(x,y)W_{x+y}.
    \end{equation}
    Taking the complex conjugate gives
    \begin{equation}
        \overline{W_xW_y}
        =
        \overline{\gamma(x,y)}\,\overline{W_{x+y}}.
    \end{equation}
    Hence the phase cancels in the tensor product:
    \begin{equation}
        (W_x\otimes \overline{W_x})(W_y\otimes \overline{W_y})
        =
        W_{x+y}\otimes \overline{W_{x+y}}.
    \end{equation}
    Therefore,
    \begin{align}
        P_A^2
        &=
        \frac{1}{|A|^2}
        \sum_{x,y\in A}
        (W_x\otimes \overline{W_x})(W_y\otimes \overline{W_y}) \\
        &=
        \frac{1}{|A|^2}
        \sum_{x,y\in A}
        W_{x+y}\otimes \overline{W_{x+y}} \\
        &=
        \frac{1}{|A|}
        \sum_{z\in A} W_z\otimes \overline{W_z} \\
        &=P_A.
    \end{align}
    Since each qubit Weyl operator $W_x$ is Hermitian, $W_x\otimes\overline{W_x}$ is
    also Hermitian. Hence $P_A^\dagger=P_A$, and so $P_A$ is an orthogonal projector.

    Next, we relate this projector to the Pauli weights. For any subspace $A\subset V$,
    \begin{align}
        \mathrm{tr}\!\left[(\rho\otimes\overline{\rho})P_A\right]
        &=
        \frac{1}{|A|}\sum_{x\in A}
        \mathrm{tr}\!\left[(\rho\otimes\overline{\rho})
        (W_x\otimes \overline{W_x})\right] \\
        &=
        \frac{1}{|A|}\sum_{x\in A}
        \mathrm{tr}(\rho W_x)\,
        \mathrm{tr}(\overline{\rho}\,\overline{W_x}) \\
        &=
        \frac{1}{|A|}\sum_{x\in A}
        \left|\mathrm{tr}(\rho W_x)\right|^2.
    \end{align}
    Since both $\rho$ and $W_x$ are Hermitian, $\mathrm{tr}(\rho W_x)$ is real. Thus
    \begin{equation}
        \mathrm{tr}\!\left[(\rho\otimes\overline{\rho})P_A\right]
        =
        \frac{1}{|A|}\sum_{x\in A}\mathrm{tr}(\rho W_x)^2.
    \end{equation}
    Applying this identity to $A=S_i$, the assumptions become
    \begin{equation}
        \mathrm{tr}\!\left[(\rho\otimes\overline{\rho})P_{S_i}\right]
        \ge 1-\epsilon_i,
        \qquad i\in\{1,2\}.
    \end{equation}

    We now compute the product $P_{S_1}P_{S_2}$. Using the same phase-cancellation
    identity as above,
    \begin{align}
        P_{S_1}P_{S_2}
        &=
        \frac{1}{|S_1||S_2|}
        \sum_{x\in S_1,\;y\in S_2}
        (W_x\otimes \overline{W_x})(W_y\otimes \overline{W_y}) \\
        &=
        \frac{1}{|S_1||S_2|}
        \sum_{x\in S_1,\;y\in S_2}
        W_{x+y}\otimes \overline{W_{x+y}} \\
        &=
        \frac{1}{|S_1||S_2|}
        \sum_{s\in S} N_s\, W_s\otimes \overline{W_s},
    \end{align}
    where
    \begin{equation}
        N_s
        \coloneq
        \left|\{(x,y)\in S_1\times S_2:x+y=s\}\right|.
    \end{equation}
    For every $s\in S=S_1+S_2$, we have $N_s=|S_1\cap S_2|$. Indeed, if
    $(x_0,y_0)$ is one solution to $x+y=s$, then all solutions are exactly
    \begin{equation}
        (x_0+u,y_0+u),
        \qquad u\in S_1\cap S_2.
    \end{equation}
    Therefore,
    \begin{align}
        P_{S_1}P_{S_2}
        &=
        \frac{|S_1\cap S_2|}{|S_1||S_2|}
        \sum_{s\in S} W_s\otimes \overline{W_s} \\
        &=
        \frac{1}{|S|}
        \sum_{s\in S} W_s\otimes \overline{W_s} \\
        &=P_S,
    \end{align}
    where we used
    \begin{equation}
        |S|=|S_1+S_2|
        =
        \frac{|S_1||S_2|}{|S_1\cap S_2|}.
    \end{equation}
    In particular, $P_{S_1}$ and $P_{S_2}$ commute.

    Since $P_{S_1}$ and $P_{S_2}$ are commuting orthogonal projectors, they are
    simultaneously diagonalizable. Thus there exists an orthonormal basis
    $\{\ket{\psi_j}\}_j$ such that
    \begin{align}
        P_{S_1}&=\sum_j a_j\ket{\psi_j}\bra{\psi_j},\\
        P_{S_2}&=\sum_j b_j\ket{\psi_j}\bra{\psi_j},\\
        P_S=P_{S_1}P_{S_2}
        &=\sum_j a_jb_j\ket{\psi_j}\bra{\psi_j},
    \end{align}
    where $a_j,b_j\in\{0,1\}$. Since
    \begin{equation}
        a_jb_j\ge a_j+b_j-1
    \end{equation}
    for all $j$, we obtain the operator inequality
    \begin{equation}
        P_S\succeq P_{S_1}+P_{S_2}-I,
    \end{equation}
    where $A \succeq B$ means that $A-B$ is positive semidefinite. Finally, since $\rho\otimes\overline{\rho}$ is positive semidefinite, taking the
    trace against it preserves the semidefinite order. Hence
    \begin{align}
        \mathrm{tr}\!\left[(\rho\otimes\overline{\rho})P_S\right]
        &\ge
        \mathrm{tr}\!\left[(\rho\otimes\overline{\rho})P_{S_1}\right]
        +
        \mathrm{tr}\!\left[(\rho\otimes\overline{\rho})P_{S_2}\right]
        -
        \mathrm{tr}(\rho\otimes\overline{\rho}) \\
        &\ge
        (1-\epsilon_1)+(1-\epsilon_2)-1 \\
        &=
        1-(\epsilon_1+\epsilon_2).
    \end{align}
    Using again
    \begin{equation}
        \mathrm{tr}\!\left[(\rho\otimes\overline{\rho})P_S\right]
        =
        \frac{1}{|S|}\sum_{x\in S}\mathrm{tr}(\rho W_x)^2,
    \end{equation}
    we obtain
    \begin{equation}
        \frac{1}{|S|}\sum_{x\in S}\mathrm{tr}(\rho W_x)^2
        \ge 1-(\epsilon_1+\epsilon_2)
    \end{equation}
    as desired.
\end{proof}
\begin{corollary}\label{cor:chain_of_SOS}
    For $\epsilon\in(0,1)$ and $i \in [n]$, let $S_i\subset V\coloneq\mathbb{F}_2^{2n}$ be the subspace such that
    \begin{equation}
        \frac{1}{|S_i|}\sum_{x \in S_i} \mathrm{tr}(\rho W_x)^2 \geq 
        1-\epsilon_i.
    \end{equation}
    Define $S=\sum_i S_i$, then
    \begin{equation}
        \frac{1}{|S|}\sum_{x\in S}\mathrm{tr}(\rho W_x)^2 
        \geq
        1-\sum_i \epsilon_i
    \end{equation}
\end{corollary}
\begin{proof}
    Define $T_m:=\sum_{i=1}^m S_i$ for $m\in[n]$. We prove by induction on $m$ that
    \begin{equation}
        \frac{1}{|T_m|}\sum_{x\in T_m}\mathrm{tr}(\rho W_x)^2 \ge 1-\sum_{i=1}^m \epsilon_i.
    \end{equation}
    The case $m=1$ is exactly the assumption. For the induction step, assume the claim holds for $T_m$. Since $T_{m+1}=T_m+S_{m+1}$, Lemma~\ref{lem:sos} implies
    \begin{equation}
        \frac{1}{|T_{m+1}|}\sum_{x\in T_{m+1}}\mathrm{tr}(\rho W_x)^2
        \ge \left(1-\sum_{i=1}^m \epsilon_i\right)-\epsilon_{m+1}
        =1-\sum_{i=1}^{m+1}\epsilon_i.
    \end{equation}
    Taking $m=n$ proves the claim.
\end{proof}

\begin{fact}[Lemma 9.4 in~\cite{grewal2025efficient}]\label{fact:reduce_2n}
    Let $S_1, \ldots, S_m$ denote $m$ (potentially non-unique) subspaces of $\mathbb{F}_2^{2n}$ and let
    \[
    S = \sum_{i=1}^m S_i.
    \]
    Suppose that for some $\varepsilon > 0$ and for all $S_i$,
    \[
    \sum_{x \in S_i} p_\rho(x) \ge \left(1 - \frac{\varepsilon}{2n}\right)\frac{|S_i|}{2^n}.
    \]
    Then,
    \[
    \sum_{x \in S} p_\rho(x) \ge (1 - \varepsilon)\frac{|S|}{2^n}.
    \]
\end{fact}

\begin{fact}[Fact 2.21 in~\cite{grewal2025efficient}]
    $M=\left\{x\in \mathbb{F}_2^{2n} : 2^n p_\rho(x) > \tfrac12 \right\}$ is isotropic.
\end{fact}

\begin{fact}[Lemma B.7 in~\cite{grewal2025efficient}]\label{fact:heavy-subspace-in-M}
    Let $M=\left\{x\in \mathbb{F}_2^{2n} : 2^n p_\rho(x) > \tfrac12 \right\}$. If $S\subseteq \mathbb{F}_2^{2n}$ is a subspace such that
    \begin{equation}
        \sum_{x\in S} p_\rho(x) > \frac{3|S|}{4\cdot 2^n},
    \end{equation}
    then $S \subseteq \langle M\rangle$, and hence $S$ is isotropic.
\end{fact}

\begin{corollary}[Correctness of the the Algorithm~\ref{alg}]\label{cor:part2_main}
    Assume that $t=O(\log n)$, $k=\Omega(\log n)$, and $\epsilon\in (0, 0.25)$.
    Let $C_1,\dots,C_m$ be sampled independently from $\mathcal{C}_k$, where $m=O(n2^t)$,
    and define $\rho_i:=C_i\rho C_i^\dagger$ for each $i\in[m]$.
    For each $i\in[m]$, let $\widehat{H}_i\subseteq \mathcal{X}$ be the subspace spanned by
    $O(n^2/\epsilon)$ independent computational difference samples drawn from $r_{\rho_i}$.
    Define
    \begin{equation}
        \widehat{S}
        :=
        \sum_{i=1}^{m} C_i^\dagger\bigl((\widehat{H}_i+\mathcal{Z})^\perp\bigr).
    \end{equation}
    Then, with high probability over the choice of $C_1,\dots,C_m$ and the computational
    difference samples, the following holds for all but a $2^{-\Omega(n)}$ fraction of
    $\mathrm{Weyl}(\rho)\in\mathcal{I}_{n-t}(V)$:
    \begin{enumerate}
        \item $\widehat{S}\supseteq \mathrm{Weyl}(\rho)$,
        \item $\displaystyle \mathbb{E}_{x\sim \widehat{S}}\!\left[\mathrm{tr}(\rho W_x)^2\right]\ge 1-\epsilon$,
        \item $\widehat{S}$ is isotropic.
    \end{enumerate}
    In particular, $\widehat{S}$ is a valid output of Algorithm~\ref{alg}.
\end{corollary}

\begin{proof}
    By Corollary~\ref{cor:weyl-support-recovery}, with high probability over the choice of
    $C_1,\dots,C_m$, we have
    \begin{equation}\label{eq:part2_support_sum}
        \sum_{i=1}^{m}\bigl(\mathrm{Weyl}(\rho)\cap C_i^\dagger(\mathcal{Z})\bigr)
        =
        \mathrm{Weyl}(\rho)
    \end{equation}
    for all but a $2^{-\Omega(n)}$ fraction of $\mathrm{Weyl}(\rho)\in\mathcal{I}_{n-t}(V)$.

    Fix such a $\rho$, and for each $i\in[m]$ define
    \begin{equation}
        H_i:=\bigl(\mathrm{Weyl}(\rho_i)\cap \mathcal{Z}\bigr)^\perp\cap \mathcal{X}.
    \end{equation}
    By Corollary~\ref{prop:cds-support}, every computational difference sample drawn from
    $r_{\rho_i}$ lies in $H_i$. Since $H_i$ is a subspace, it follows that
    \begin{equation}
        \widehat{H}_i\subseteq H_i.
    \end{equation}
    Therefore,
    \begin{equation}
        (\widehat{H}_i+\mathcal{Z})^\perp \supseteq (H_i+\mathcal{Z})^\perp.
    \end{equation}
    By Lemma~\ref{lem:cds-orthogonal},
    \begin{equation}
        (H_i+\mathcal{Z})^\perp=\mathrm{Weyl}(\rho_i)\cap \mathcal{Z}.
    \end{equation}
    Conjugating by $C_i^\dagger$, we obtain
    \begin{equation}
        C_i^\dagger\bigl((\widehat{H}_i+\mathcal{Z})^\perp\bigr)
        \supseteq
        C_i^\dagger\bigl(\mathrm{Weyl}(\rho_i)\cap \mathcal{Z}\bigr)
        =
        \mathrm{Weyl}(\rho)\cap C_i^\dagger(\mathcal{Z}).
    \end{equation}
    Summing over $i$ and using Eq.~\eqref{eq:part2_support_sum}, we conclude that
    \begin{equation}
        \widehat{S}\supseteq \mathrm{Weyl}(\rho).
    \end{equation}
    Next, applying Corollary~\ref{cor:r_mass} to each $\rho_i$ with error parameter
    $\epsilon/(2n)$ and failure probability $\delta/m$, it suffices to take
    \begin{equation}
        N_S
        =
        O\!\left(
            \frac{n(n+\log(m/\delta))}{\epsilon}
        \right)
    \end{equation}
    computational difference samples for each $i\in[m]$. Then, by a union bound, with
    probability at least $1-\delta$,
    \begin{equation}
        r_{\rho_i}(\widehat{H}_i)\ge 1-\frac{\epsilon}{2n}
    \end{equation}
    holds simultaneously for all $i\in[m]$.

    Define
    \begin{equation}
        \widehat{S}_i:=C_i^\dagger\bigl((\widehat{H}_i+\mathcal{Z})^\perp\bigr).
    \end{equation}
    For each $i$, using $p_{C\rho C^{\dagger}}(S)=p_{\rho}(C^{\dagger}(S))$ and
    Lemma~\ref{lem:cds-subspace-mass}, we have
    \begin{align}
        p_\rho(\widehat{S}_i)
        &=
        p_{\rho_i}((\widehat{H}_i+\mathcal{Z})^\perp) \\
        &=
        \frac{|(\widehat{H}_i+\mathcal{Z})^\perp|}{2^n}\, r_{\rho_i}(\widehat{H}_i) \\
        &\ge
        \left(1-\frac{\epsilon}{2n}\right)\frac{|(\widehat{H}_i+\mathcal{Z})^\perp|}{2^n} \\
        &=
        \left(1-\frac{\epsilon}{2n}\right)\frac{|\widehat{S}_i|}{2^n}.
    \end{align}
    Applying Fact~\ref{fact:reduce_2n} to the family $\{\widehat{S}_i\}_{i=1}^m$, we obtain
    \begin{equation}
        \sum_{x\in \widehat{S}} p_\rho(x)
        \ge
        (1-\epsilon)\frac{|\widehat{S}|}{2^n}.
    \end{equation}
    Since $p_\rho(x)=2^{-n}\mathrm{tr}(\rho W_x)^2$, this is equivalent to
    \begin{equation}
        \mathbb{E}_{x\sim \widehat{S}}\!\left[\mathrm{tr}(\rho W_x)^2\right]\ge 1-\epsilon.
    \end{equation}

    Finally, because $\epsilon<1/4$, we have
    \begin{equation}
        \sum_{x\in \widehat{S}} p_\rho(x)
        >
        \frac{3|\widehat{S}|}{4\cdot 2^n}.
    \end{equation}
    Therefore, Fact~\ref{fact:heavy-subspace-in-M} implies that $\widehat{S}$ is isotropic.
\end{proof}

\section{Worst-case instance: GHZ}\label{appx:ghz}

Our learning algorithm does not succeed for every input whose stabilizer group $\mathrm{Weyl}(\rho)$ lies in $\mathcal{I}_{n-t}(V)$. As noted in the main text, the $n$-qubit GHZ state
provides a representative hard instance. Let
\[
    \rho_{\mathrm{GHZ}} = \ket{\mathrm{GHZ}}\!\bra{\mathrm{GHZ}},
\]
where $\ket{\mathrm{GHZ}}=\frac{1}{\sqrt{2}}(\ket{0^n}+\ket{1^n})$ and let $S_{\mathrm{GHZ}}$ denote its stabilizer group, generated by
\[
    S_{\mathrm{GHZ}}
    =
    \left\langle
    Z_1 Z_2,\,
    Z_2 Z_3,\,
    \dots,\,
    Z_{n-1} Z_n,\,
    \bigotimes_{i=1}^n X_i
    \right\rangle.
\]
Define
\[
    T_0
    :=
    \left\langle
    Z_1 Z_2,\,
    Z_2 Z_3,\,
    \dots,\,
    Z_{n-1} Z_n
    \right\rangle
    =
    \{(0^n,z)\in \mathbb{F}_2^{2n} : \textstyle\sum_i z_i \equiv 0 \pmod 2\}.
\]
Then
\[
    S_{\mathrm{GHZ}} \setminus T_0
    =
    \{(1^n,z)\in \mathbb{F}_2^{2n} : \textstyle\sum_i z_i \equiv 0 \pmod 2\}.
\]
If $C$ is sampled uniformly from $\mathcal{C}_k$, then we obtain
\begin{equation}\label{eq:ghz_pr_upper}
    \Pr_C\!\left(
        C(S_{\mathrm{GHZ}}\setminus T_0)\cap \mathcal{Z}\neq \emptyset
    \right)
    \le (2/3)^{n/k}.
\end{equation}
This shows that, under the ensemble $\mathcal{C}_k$, one would need exponentially many
independent random Clifford circuits in order to reveal all generators in $S_{\mathrm{GHZ}}$.

There are two points worth emphasizing about this observation. First, the bound in
Eq.~\eqref{eq:ghz_pr_upper} is proved only for the Clifford ensemble $\mathcal{C}_k=\mathrm{Cl}(k)^{\otimes n/k}$ used throughout this paper. It does not rule out
the possibility that a different choice of shallow-depth Clifford ensemble could avoid this
exponential suppression. Understanding whether such an ensemble exists remains an interesting direction for future work.

Second, even within the ensemble $\mathcal{C}_k$, the GHZ example does not by itself imply
that shallow-depth measurements are fundamentally insufficient. As discussed in the main
text, once one allows a single adaptive step, the full stabilizer group of the GHZ state can
still be learned using $O(\log n)$-depth measurements without any exponential increase in the
number of samples. This suggests that adaptivity may substantially enlarge the class of states
that can be learned efficiently in the shallow-depth setting. Determining how much adaptivity
is necessary remains another interesting open question.

We now introduce the lemma needed to prove Eq.~\eqref{eq:ghz_pr_upper}.

\begin{lemma}\label{lem:ghz_pr_2/3}
    Let $x_0=(1^n,0^n)\in V\coloneq \mathbb{F}_2^{2n}$, and let $C$ be sampled uniformly from $\mathrm{Cl}(n)$. Then
    \begin{equation}
        \Pr_C\!\bigl(C(x_0+\mathcal{Z})\cap \mathcal{Z}\neq \emptyset\bigr)=\frac{2}{3}.
    \end{equation}
\end{lemma}

\begin{proof}
    We first replace the random Clifford $C$ by the random Lagrangian subspace
    $L:=C^\dagger(\mathcal{Z})$. Since $C$ is sampled uniformly from $\mathrm{Cl}(n)$, the
    subspace $L$ is distributed uniformly over the set of Lagrangian subspaces of $V$.
    Therefore,
    \begin{align}
        \Pr_C\!\bigl(C(x_0+\mathcal{Z})\cap \mathcal{Z}\neq \emptyset\bigr)
        &= \Pr_C\!\bigl((x_0+\mathcal{Z})\cap C^\dagger(\mathcal{Z})\neq \emptyset\bigr) \\
        &= \Pr_L\!\bigl((x_0+\mathcal{Z})\cap L\neq \emptyset\bigr) \\
        &= \Pr_L\!\bigl(x_0\in \mathcal{Z}+L\bigr).
        \label{eq:cond_lag}
    \end{align}

    Next, consider the quotient map
    \begin{equation}
        \pi:V\to V/\mathcal{Z},\qquad \pi(x)=x+\mathcal{Z}.
    \end{equation}
    We claim that
    \begin{equation}
        x_0\in \mathcal{Z}+L
        \quad\Longleftrightarrow\quad
        \pi(x_0)\in \pi(L),
    \end{equation}
    where $\pi(L)=\{l+\mathcal{Z}: l\in L\}$. Indeed, if $x_0\in \mathcal{Z}+L$, then $x_0=z+l$ for some $z\in \mathcal{Z}$ and
    $l\in L$, so $\pi(x_0)=x_0+\mathcal{Z}=l+\mathcal{Z}\in \pi(L)$. Conversely, if $\pi(x_0)\in \pi(L)$, then $x_0+\mathcal{Z}=l+\mathcal{Z}$ for some
    $l\in L$, which means that $x_0-l\in \mathcal{Z}$, or equivalently
    $x_0\in \mathcal{Z}+L$.

    Thus,
    \begin{align}
        \Pr_L\!\bigl(x_0\in \mathcal{Z}+L\bigr)
        &= \Pr_L\!\bigl(\pi(x_0)\in \pi(L)\bigr) \\
        &= \sum_{r=0}^{n}
        \Pr_L\!\bigl(\dim(L\cap \mathcal{Z})=r\bigr)\,
        \Pr_L\!\bigl(\pi(x_0)\in \pi(L)\mid \dim(L\cap \mathcal{Z})=r\bigr).
    \end{align}

    By Fact~\ref{fact:num_lag_intersection},
    \begin{equation}
        \Pr_L\!\bigl(\dim(L\cap \mathcal{Z})=r\bigr)
        =
        \frac{\binom{n}{r}_2\,2^{(n-r)(n-r+1)/2}}{\prod_{i=1}^{n}(2^i+1)}.
    \end{equation}
    Moreover, conditioned on $\dim(L\cap \mathcal{Z})=r$, the image $\pi(L)$ is uniformly distributed over all $(n-r)$-dimensional subspaces of $V/\mathcal{Z}$, since the subgroup $\{\phi\in\mathrm{Sp}(V):\phi(\mathcal{Z})=\mathcal{Z}\}$ acts transitively on that set. Therefore,
    since $\pi(x_0)\neq 0$, we have
    \begin{equation}
        \Pr_L\!\bigl(\pi(x_0)\in \pi(L)\mid \dim(L\cap \mathcal{Z})=r\bigr)
        =
        \frac{\binom{n-1}{n-r-1}_2}{\binom{n}{n-r}_2}
        =
        \frac{2^{\,n-r}-1}{2^n-1}.
    \end{equation}
    Hence,
    \begin{align}
        \Pr_L\!\bigl(x_0\in \mathcal{Z}+L\bigr)
        &=
        \sum_{r=0}^{n}
        \frac{\binom{n}{r}_2\,2^{(n-r)(n-r+1)/2}}{\prod_{i=1}^{n}(2^i+1)}
        \cdot
        \frac{2^{\,n-r}-1}{2^n-1} \\
        &=
        \frac{2^{n+1}+1}{3(2^n-1)}-\frac{1}{2^n-1} \\
        &= \frac{2}{3}
    \end{align}
    as desired.
\end{proof}

We can now prove Eq.~\eqref{eq:ghz_pr_upper}.

\begin{proof}[Proof of Eq.~\eqref{eq:ghz_pr_upper}]
    Since
    \[
        S_{\mathrm{GHZ}}\setminus T_0 = x_0+T_0 \subseteq x_0+\mathcal{Z},
    \]
    we obtain
    \begin{align}
        \Pr_C\!\bigl(C(S_{\mathrm{GHZ}}\setminus T_0)\cap \mathcal{Z}\neq \emptyset\bigr)
        &= \Pr_C\!\bigl(C(x_0+T_0)\cap \mathcal{Z}\neq \emptyset\bigr) \\
        &\le \Pr_C\!\bigl(C(x_0+\mathcal{Z})\cap \mathcal{Z}\neq \emptyset\bigr).
    \end{align}

    Now let $C=C_1\otimes\cdots\otimes C_{n/k}$ be sampled uniformly from
    $\mathcal{C}_k=\mathrm{Cl}(k)^{\otimes n/k}$. Since $x_0+\mathcal{Z}$ factorizes across the
    $n/k$ blocks as
    \[
        x_0+\mathcal{Z}
        =
        (x_{0,k}+\mathcal{Z}_k)^{\otimes n/k},
    \]
    where
    \[
        x_{0,k}=(1^k,0^k)\in \mathbb{F}_2^{2k},
        \qquad
        \mathcal{Z}_k=\{I,Z\}^k=\{(0^k,z)\in \mathbb{F}_2^{2k}: z\in \mathbb{F}_2^k\},
    \]
    the event factorizes over $n/k$ blocks as well. Therefore,
    \begin{align}
        \Pr_C\!\bigl(C(x_0+\mathcal{Z})\cap \mathcal{Z}\neq \emptyset\bigr)
        &= \prod_{i=1}^{n/k}
        \Pr_{C_i}\!\bigl(C_i(x_{0,k}+\mathcal{Z}_k)\cap \mathcal{Z}_k\neq \emptyset\bigr) \\
        &= \left(\frac{2}{3}\right)^{n/k},
    \end{align}
    where the last equality follows from Lemma~\ref{lem:ghz_pr_2/3} applied on each
    $k$-qubit block. This proves Eq.~\eqref{eq:ghz_pr_upper}.
\end{proof}

\section{Numerical simulations}\label{appx:numeric}

In this section, we provide additional details on the numerical simulations presented in the main text.
Our theoretical analysis shows that, for the Clifford ensemble $\mathcal{C}_k$ with $k=\Omega(\log n)$, the proposed learning procedure succeeds for all but an exponentially small $2^{-\Omega(n)}$ fraction of $(n-t)$-dimensional stabilizer groups $\mathrm{Weyl}(\rho)\in \mathcal{I}_{n-t}(V)$, with the same asymptotic guarantee as in the fully random Clifford measurements $(k=n)$.
Moreover, our numerical results suggest that, for the random instances considered here, even the case $k=1$, corresponding to single-qubit measurements, already performs comparably well in practice.

In Fig.~\ref{fig:2}(a) of the main text, we estimate the number of Clifford circuits required to learn the generators of $\mathrm{Weyl}(\rho)$ for the two Clifford ensembles $\mathcal{C}_1$ and $\mathcal{C}_n$.
We fix $n=100$ and vary $t=0,\dots,5$.
For each value of $t$, we sample $1000$ random isotropic subspaces of $\mathbb{F}_2^{2n}$ of dimension $n-t$.
For each sampled subspace $S=\mathrm{Weyl}(\rho)$, we sequentially draw Clifford circuits $C_1,C_2,\dots, C_m$ independently from each Clifford ensemble, and define
\begin{equation}
    S_i := S \cap C_i^\dagger(\mathcal{Z}).
\end{equation}
Starting from $A_0=\{0\}$, we recursively set
\begin{equation}
    A_i := \langle S_1,\dots,S_i\rangle.
\end{equation}
We then record the smallest integer $m$ such that $A_m=S$.
Thus, $m$ is the number of sampled Clifford circuits required until the accumulated generators span the full hidden symmetry group. 

\begin{figure}[t]
    \centering
    \includegraphics[width=0.4\linewidth, trim=1cm 8.5cm 18cm 1cm, clip]{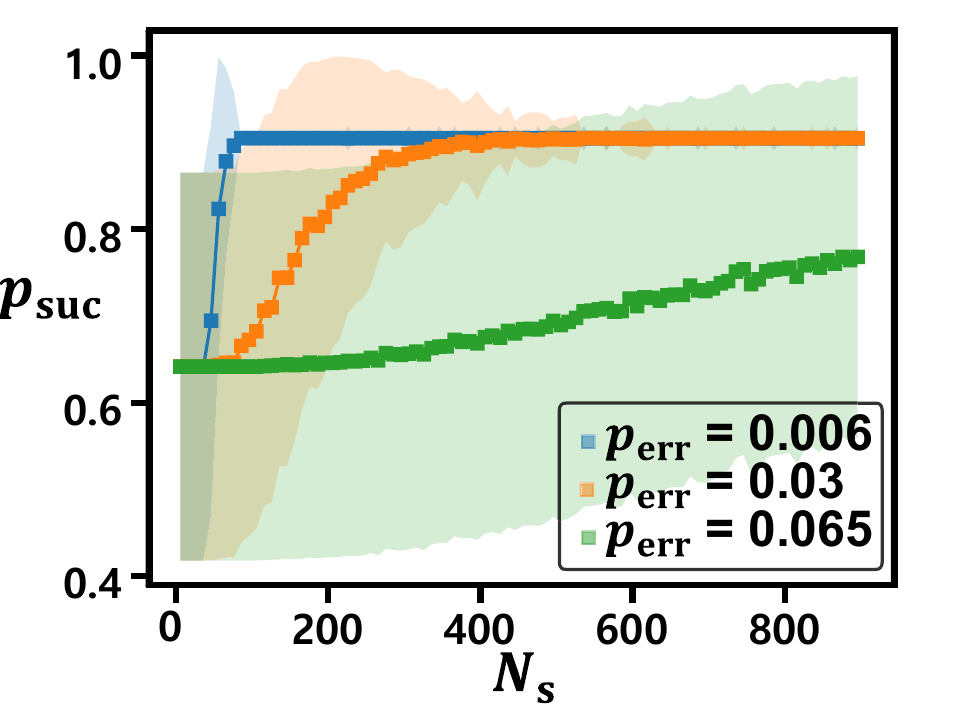}
    \caption{\label{fig_2b_std} Recovery under noisy computational difference samples. We fix $n=24$, $t=0$, and vary the number $N_S$ of computational difference samples obtained from a single measurement basis. We evaluate the success probability $p_{\mathrm{suc}}$ of recovering the hidden symmetry when each measured bit is flipped independently with probability $p_{\mathrm{err}}$. Each point shows the mean, and the shaded region indicates one standard deviation over $1000$ independent trials.
    }
\end{figure}

In Fig.~\ref{fig:2}(b) of the main text, we study recovery of a hidden Pauli symmetry from noisy samples.
Whenever a Clifford circuit $C$ has a nontrivial overlap with $\mathrm{Weyl}(\rho)$, the resulting post-processing problem reduces to estimating a hidden parity constraint.
More precisely, there exists a secret string $s\in\mathbb{F}_2^n$ such that, in the noiseless case, the samples satisfy
\begin{equation}
    s\cdot x \equiv 0 \pmod 2.
\end{equation}
In the noisy setting, this condition is only satisfied with a bias determined by the sample error rate $p_{\mathrm{err}}$.

To recover the hidden parity, we use the Fast Walsh--Hadamard Transform (FWHT), which computes all empirical Walsh correlations once the sample frequency is estimated.
Given samples $\{x_1,\dots,x_{N_S}\}\subseteq\{0,1\}^n$, let $\widehat{f}$ denote the empirical frequency, and define
\begin{equation}
    c(s)
    :=
    \sum_{x\in\{0,1\}^n} \widehat{f}(x)(-1)^{s\cdot x},
    \qquad s\in\{0,1\}^n.
\end{equation}
In our implementation, we exclude the trivial coefficient $c(0)$ and use the estimator
\begin{equation}
    \widehat{s}
    :=
    \mathrm{argmax}_{s\in\{0,1\}^n\setminus\{0^n\}} c(s).
\end{equation}
Since in our simulation the hidden parity is biased toward the event $s\cdot x=0$, the correct string $s$ yields a large positive Walsh coefficient. Therefore, in this setting it suffices to maximize $c(s)$.

To suppress false positives when no nontrivial hidden symmetry is present, we introduce a detection threshold
\begin{equation}
    \gamma
    :=
    \sqrt{\frac{2\bigl(n\log 2 + \log(1/\alpha)\bigr)}{N_S}},
\end{equation}
where $\alpha$ is the target confidence level.
If
\begin{equation}
    \max_{s\neq 0} c(s) \le \gamma,
\end{equation}
we declare that the corresponding Clifford circuit $C$ does not reveal a detectable nontrivial symmetry at that step.
This choice is motivated by a union-bound estimate over all nonzero Walsh coefficients, which ensures that the false positive probability under the null model is at most $\alpha$.

In the experiment, we fix $n=24$, vary the number of samples $N_S$, and evaluate the success probability of recovering the hidden symmetry for several values of the sample error rate $p_{\mathrm{err}}$. In the main text, the standard-deviation bars are omitted for visual clarity; the corresponding plot including them is shown in Fig.~\ref{fig_2b_std}. Each plotted point is averaged over 1000 independent trials.
If the hidden component is trivial, a trial is counted as successful only when the algorithm outputs the trivial estimate.
If the hidden component is nontrivial, a trial is counted as successful whenever the returned vector is consistent with the target hidden symmetry.
Throughout this experiment, we assume that, conditioned on the hidden constraint, the valid samples are drawn uniformly. The FWHT-based procedure is suitable for moderate system sizes, but it becomes impractical as $n$ grows, since it requires evaluating correlation scores over all $2^n$ candidate bit strings. As a result, both the samples and the post-processing cost scale exponentially with $n$. For larger systems, one may instead use the Goldreich--Levin algorithm~\cite{goldreich1989hardcore}, which recovers heavy Fourier coefficients without explicitly scanning the full $\mathbb{F}_2^{n}$.

\begin{figure}[t]
    \centering
    \includegraphics[width=0.8\linewidth, trim=0cm 3.5cm 11cm 0cm, clip]{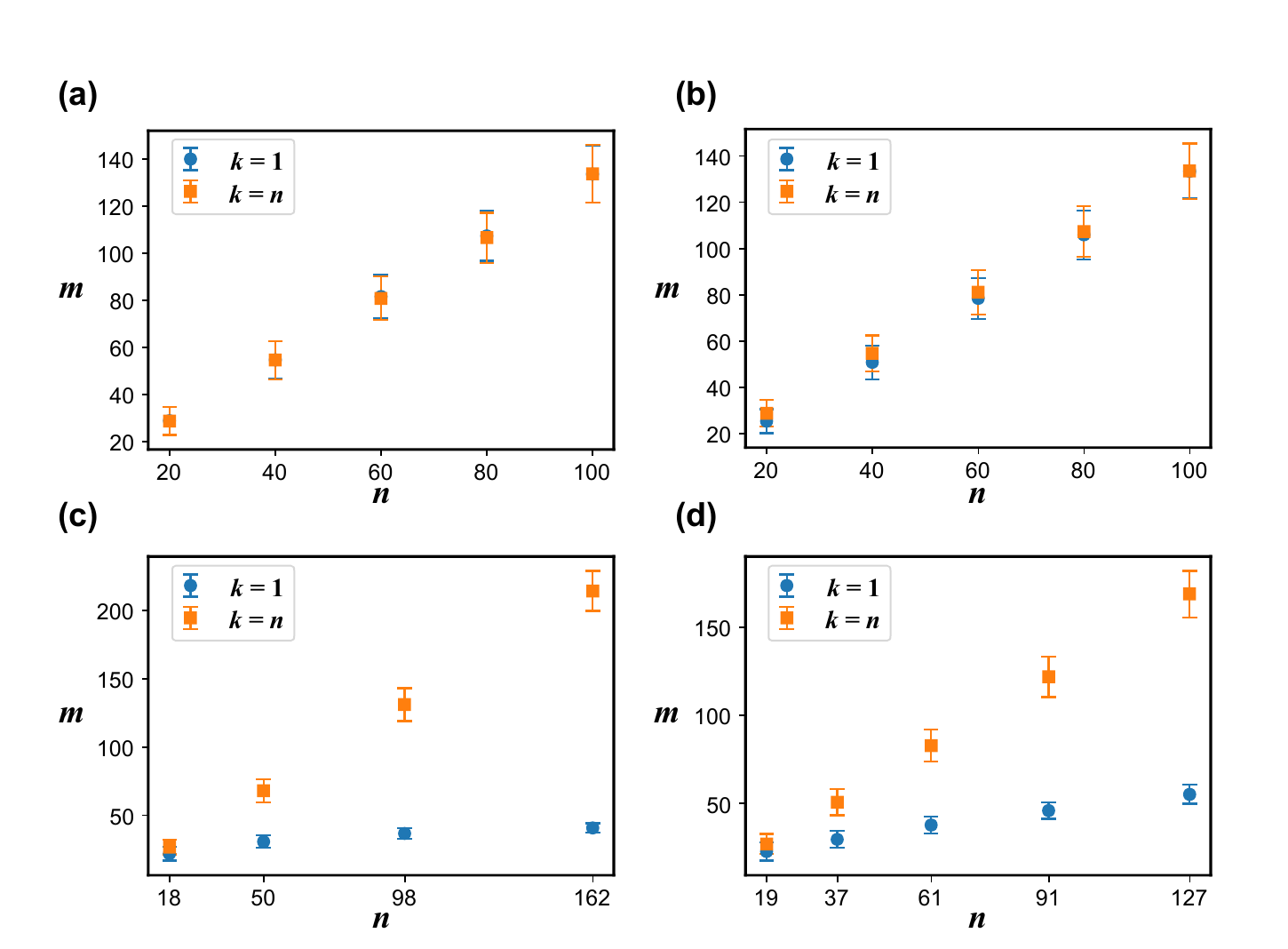}
    \caption{\label{fig:5}
        Numerical simulations for learning stabilizer groups with different Clifford ensembles.
        We compare the number $m$ of sampled Clifford circuits required to recover the stabilizer
        group $\mathrm{Weyl}(\rho)$ using $\mathcal{C}_k=\mathrm{Cl}(k)^{\otimes n/k}$ for $k=1$ and $k=n$.
        (a) Random stabilizer states, as a function of the system size $n$.
        (b) Random CSS-code states, where $\rho$ is taken to be a logical state of a random CSS code.
        (c) Toric-code logical states, plotted as a function of the number of physical qubits $n$
        for different code distances.
        (d) Color-code logical states, plotted as a function of the number of physical qubits $n$
        for different code distances.
        Each point shows the mean and standard deviation over $1000$ independent trials.
        }
\end{figure}
In addition to the experiments reported in the main text, we performed several supplementary numerical simulations that further support our findings [Fig.~\ref{fig:5}].
First, we compared the performance of different Clifford ensembles as a function of the system size $n$.
For $n=20,40,\dots,100$, we estimated the number of Clifford circuits $m$ required to learn $\mathrm{Weyl}(\rho)$.
Each plotted point reports the empirical mean and standard deviation (std) over $1000$ random instances.
Second, we repeated the same comparison not only for uniformly random stabilizer groups, but also for random CSS codes, which arise more frequently in quantum error correction.
Again, for $n=20,40,\dots,100$, we estimated the number of Clifford circuits $m$ required to reconstruct $\mathrm{Weyl}(\rho)$ and compared the resulting performance across different Clifford ensembles.
Third, we considered structured, non-random instances given by a single logical state of the toric code~\cite{kitaev2003fault} and of the color code~\cite{bombin2006topological}.
For both families, we varied the code distance over
\begin{align}
    &d_{\mathrm{code}}^{\mathrm{(color)}} \in \{5,7,9,11,13\},\\
    &d_{\mathrm{code}}^{\mathrm{(toric)}} \in \{3,5,7,9\}.
\end{align}
and compared the performance of different Clifford ensembles in terms of the number of circuits required for successful recovery.
These experiments indicate that the qualitative behavior observed for shallow Clifford ensembles is not restricted to uniformly random stabilizer instances, but also persists for physically motivated code states. Interestingly, in these examples, $\mathcal{C}_1$ even outperforms $\mathcal{C}_n$. This is because logical states of commonly studied quantum error-correcting codes are typically generated by more structured and local entangling circuits than highly nonlocal states such as GHZ.

\section{Lower bound on the $t$-dependence}\label{appx:lower}

In this section, we show that the exponential dependence on $t$ is unavoidable for
single-copy learning of Pauli symmetries. More precisely, even if adaptive
single-copy measurements are allowed, any algorithm that approximates
$\mathrm{Weyl}(\rho)$ in the sense of Definition~\ref{def:learning_task}
requires $\Omega(2^t)$ copies in the worst case.

We consider algorithms $\mathcal{A}$ that, given an $n$-qubit state $\rho$, output a subspace
$\widehat S \subset \mathbb{F}_2^{2n}$ satisfying
\begin{flalign}
    &\quad 1.\ \widehat S \supseteq \mathrm{Weyl}(\rho), \\
    &\quad 2.\ \frac{1}{|\widehat S|}\sum_{x\in \widehat S}\mathrm{tr}(\rho W_x)^2 \ge 1-\epsilon, \\
    &\quad 3.\ \widehat S \text{ is isotropic}. &&
\end{flalign}

We will show that the existence of such an algorithm would imply an efficient
single-copy protocol for a previously studied hypothesis-testing problem that is known
to require $\Omega(2^t)$ samples. This yields the desired lower bound.

Consider the following two families of states:
\begin{alignat}{2}
    \rho_A
    &= \ket{0^{n-t}}\bra{0^{n-t}} \otimes \frac{I}{2^t}, 
    \qquad &\text{with probability } \frac{1}{2},\\
    \rho_B
    &= \ket{0^{n-t}}\bra{0^{n-t}} \otimes \frac{I+P}{2^t}, 
    \qquad &\text{with probability } \frac{1}{2},
\end{alignat}
where $P$ is randomly chosen from $\{I,X,Y,Z\}^{\otimes t}\setminus\{I^{\otimes t}\}$.
Then
\begin{align}
    \mathrm{Weyl}(\rho_A)
    &= \langle Z_1,\dots,Z_{n-t}\rangle, \\
    \mathrm{Weyl}(\rho_B)
    &= \langle Z_1,\dots,Z_{n-t}, P\rangle,
\end{align}
where, in the second line, $P$ is understood as acting on the last $t$ qubits.

Assume now that $\epsilon < 1/2$. Then condition 2 implies
\begin{equation}
    \frac{1}{|\widehat S|}\sum_{x\in \widehat S}\mathrm{tr}(\rho W_x)^2 > \frac12.
\end{equation}
We claim that, for both $\rho=\rho_A$ and $\rho=\rho_B$, this forces
\[
    \widehat S = \mathrm{Weyl}(\rho).
\]

Indeed, by condition 1 we already have $\widehat S \supseteq \mathrm{Weyl}(\rho)$.
Suppose, for contradiction, that $\widehat S$ strictly contains $\mathrm{Weyl}(\rho)$, and let
$x \in \widehat S\setminus \mathrm{Weyl}(\rho)$. Since $\widehat S$ is a subspace containing
$\mathrm{Weyl}(\rho)$, it also contains the entire coset $x+\mathrm{Weyl}(\rho)$.
Moreover, because $x \notin \mathrm{Weyl}(\rho)$, none of the elements of this coset belong to
$\mathrm{Weyl}(\rho)$. Therefore,
\[
    \mathrm{tr}(\rho W_y)^2 = 0
    \qquad\text{for all } y\in x+\mathrm{Weyl}(\rho),
\]
whereas
\[
    \mathrm{tr}(\rho W_y)^2 = 1
    \qquad\text{for all } y\in \mathrm{Weyl}(\rho).
\]
Since the two cosets $\mathrm{Weyl}(\rho)$ and $x+\mathrm{Weyl}(\rho)$ have the same size,
it follows that
\begin{equation}
    \frac{1}{|\widehat S|}\sum_{y\in \widehat S}\mathrm{tr}(\rho W_y)^2 \le \frac12,
\end{equation}
contradicting condition 2. Hence $\widehat S=\mathrm{Weyl}(\rho)$.

It follows that any algorithm satisfying the above three conditions can distinguish
$\rho_A$ from $\rho_B$: its output has dimension $n-t$ in the first case and $n-t+1$
in the second.

We now express such an algorithm as an adaptive single-copy measurement protocol.
Let
\[
    \{F_k^{<i}\}_k
\]
denote the POVM used on the $i$-th copy, where the choice of POVM may depend on all
previous $i-1$ outcomes. Since both $\rho_A$ and $\rho_B$ are tensor products with the same
pure state on the first $n-t$ qubits, distinguishing $\rho_A$ from $\rho_B$ is equivalent
to distinguishing their reduced states on the last $t$ qubits using the induced POVMs
\begin{equation}
    G_k^{<i}
    :=
    \mathrm{tr}_{1,\dots,n-t}\!\Big[
        \bigl(\ket{0^{n-t}}\bra{0^{n-t}} \otimes I\bigr) F_k^{<i}
    \Big].
\end{equation}
Each $\{G_k^{<i}\}_k$ is a valid POVM on $t$ qubits. Therefore the original task reduces to
the following adaptive single-copy property-testing problem on $t$ qubits:
\begin{alignat}{2}
    \rho_A'
    &= \frac{I}{2^t}, 
    \qquad &\text{with probability } \frac{1}{2},\\
    \rho_B'
    &= \frac{I+P}{2^t}, 
    \qquad &\text{with probability } \frac{1}{2}.
\end{alignat}
where $P \in \{I,X,Y,Z\}^{\otimes t}\setminus\{I^{\otimes t}\}$.

This is the single-copy Pauli-learning/testing problem studied in previous work~\cite{chen2024optimal, huang2021information, chen2022exponential},
for which an $\Omega(2^t)$ lower bound is known. Hence any algorithm satisfying conditions 1--3 above must use at least $\Omega(2^t)$ copies in the worst case. If we restrict the algorithm to be non-adaptive like ours, then we can get a tighter $\Omega(t2^t)$ lower bound~\cite{huang2021information}. 

We summarize the discussion as follows:

\begin{lemma}[Single-copy lower bound on the $t$-dependence]
    Let $\epsilon < 1/2$. Any potentially adaptive single-copy algorithm that outputs
    an isotropic subspace $\widehat S$ satisfying
    \begin{enumerate}
        \item $\widehat S \supseteq \mathrm{Weyl}(\rho)$,
        \item $\displaystyle \frac{1}{|\widehat S|}\sum_{x\in\widehat S}\mathrm{tr}(\rho W_x)^2 \ge 1-\epsilon$,
    \end{enumerate}
    for every $n$-qubit state $\rho$ with $\dim \mathrm{Weyl}(\rho)=n-t$, must use
    at least $\Omega(2^t)$ copies in the worst case.
\end{lemma}

Previous single-copy learning algorithms~\cite{grewal2025efficient, chia2024efficient} were established mainly for pure-state inputs,
so, taken in isolation, a lower bound based on mixed-state hypotheses could be viewed
as an artifact of the particular mixed-state construction used in the proof of the
lower bound. By contrast, combining the extension to mixed states in
Lemma~\ref{lem:cds-subspace-mass} with previous single-copy learning
algorithms~\cite{grewal2025efficient, chia2024efficient} shows that the relevant
parameter is the dimension of $\mathrm{Weyl}(\rho)$, rather than whether the input
state is pure or mixed. This suggests that the exponential dependence on $t$ is
intrinsic to the single-copy task of identifying hidden Pauli symmetries, rather than
being a consequence of the mixed-state instances used in the lower-bound construction.
\end{document}